\documentclass[a4paper]{article}
\usepackage[dvipsnames]{xcolor}
\usepackage[utf8]{inputenc}
\usepackage[english]{babel}
\usepackage{amsmath}
\usepackage{amssymb}
\usepackage{amsthm}
\usepackage{comment}
\usepackage{hyperref}
\usepackage{multicol}
\usepackage{tikz}
\usetikzlibrary{shapes,arrows}
\usepackage[nottoc]{tocbibind}
\usepackage{graphicx}
\usepackage{bbold}
\usepackage{csquotes}

\usepackage{enumitem}
\usepackage[T1]{fontenc}
\usepackage{lmodern}

\usepackage[
  a4paper,
  textwidth=380pt,
  textheight=632pt,
  heightrounded,
  hcentering,
  vcentering
]{geometry}
\usepackage{fancyhdr}
\usepackage{mathrsfs}
\usepackage[normalem]{ulem}

\usepackage{import}

\usepackage[capitalize]{cleveref}

\usepackage[doi=false,isbn=false, url=false, date=year, backend=biber,style=numeric,
  ]{biblatex} 
\renewbibmacro*{journal+issuetitle}{%
  \usebibmacro{journal}%
  \setunit*{\addspace}%
  \iffieldundef{series}
    {}
    {\newunit
     \printfield{series}%
     \setunit{\addspace}}%
  \usebibmacro{volume+number+eid}%
  \setunit{\bibpagespunct}%
  \printfield{pages}%
  \setunit{\addspace}%
  \usebibmacro{issue+date}%
  \setunit{\addcolon\space}%
  \usebibmacro{issue}%
  \newunit}

\renewbibmacro*{note+pages}{%
  \printfield{note}%
  \newunit}

\DeclareMathOperator{\aut}{Aut}

\newcommand{\T}{\mathbb{T}}
\newcommand{\E}{\mathbb{E}}
\newcommand{\tr}{\mathrm{tr}}
\newcommand{\e}{\mathrm{e}}
\newcommand{\D}{\mathrm{d}}

\renewcommand{\i}{\mathrm{i}}
\renewcommand{\Phi}{\varPhi}
\renewcommand{\Psi}{\varPsi}
\renewcommand{\Sigma}{\varSigma}
\newcommand{\epsi}{\varepsilon}
\newcommand{\fd}{r}

\newcommand{\g}{\gamma}
\newcommand{\La}{\Lambda}

\newcommand{\w}{\omega}
\newcommand{\N}{\mathbb{N}}
\newcommand{\Z}{\mathbb{Z}}
\newcommand{\C}{\mathbb{C}}
\newcommand{\Q}{\mathbb{Q}}

\newcommand{\mA}{\mathcal{A}}
\newcommand{\R}{\mathbb{R}}
\newcommand{\eps}{\varepsilon}

\newcommand{\OD}[1][]{   ^{\mathrm{OD}#1}   }
\newcommand{\ODT}[1][]{   ^{\widetilde{\mathrm{OD}}#1}   }
\newcommand{\mL}{\mathcal{L}}
\newcommand{\mP}{\mathcal{P}}
\newcommand{\dd}{\mathrm{d}}

\usepackage{mathtools}

\DeclarePairedDelimiter\norm{\lVert}{\rVert}
\DeclarePairedDelimiter\abs{\lvert}{\rvert}
\DeclarePairedDelimiter\br{\lparen}{\rparen}

\usepackage{thmtools} 
\usepackage{thm-restate}

\theoremstyle{plain}
\newtheorem{theorem}{Theorem}[section]
\newtheorem{lemma}[theorem]{Lemma}
\newtheorem{proposition}[theorem]{Proposition}

\theoremstyle{definition}
\newtheorem{definition}[theorem]{Definition}

\theoremstyle{remark}
\newtheorem{remark}[theorem]{Remark}

\newcommand\numberthis{\addtocounter{equation}{1}\tag{\theequation}}

\usepackage{colonequals}

\title{The Colored Hofstadter Butterfly as\\ a Many-Body Quantum Hall Phase Diagram}
\author{
    Giovanna Marcelli%
    \texorpdfstring{%
        \footnote{
            \parbox[t]{.8\textwidth}{
                \foreignlanguage{italian}{Dipartimento di Matematica e Fisica, Università di Roma Tre,\\ L.go S. L. Murialdo 1, 00146 Roma,} Italy
            }
        }
    }{}%
    \and Tadahiro Miyao%
    \texorpdfstring{%
        \footnote{
            \parbox[t]{.8\textwidth}{
                \foreignlanguage{italian}{Department of Mathematics,  Hokkaido University,
                 Sapporo 060-0810,} Japan
            }
        }
    }{}%
    \and Domenico Monaco%
    \texorpdfstring{%
        \footnote{
            \parbox[t]{.8\textwidth}{
                \foreignlanguage{italian}{Dipartimento di Matematica, 
                Sapienza Università di Roma,\\ Piazzale Aldo Moro 5, 00185 Roma, Italy}
            }
        }
    }{}%
    \and Stefan Teufel%
    \texorpdfstring{%
        \footnotemark[4]
    }{}%
    \and
     Marius Wesle%
    \texorpdfstring{\footnote{\parbox[t]{.8\textwidth}{
                \foreignlanguage{ngerman}{Fachbereich Mathematik, 
                Universität Tübingen,
                Auf~der~Morgenstelle~10,\\ 72076~Tübingen,} Germany
            }
        }
    }{}%
}

\date{\today}

\usepackage{microtype}
\begin{document}

\maketitle

\begin{abstract}

We prove that the colored Hofstadter butterfly has a many-body interpretation for a broad class of weakly interacting lattice fermion systems. Starting from a spectral gap of a Hofstadter-like one-particle Hamiltonian at arbitrary magnetic flux $b$, we construct an open region in the three-dimensional parameter space $(b,\mu,\lambda)$ of magnetic field, chemical potential, and interaction strength on which the infinite-volume interacting system has locally unique gapped ground states. The construction combines quasi-adiabatic continuation in the interaction strength with denominator-independent magnetic perturbation estimates, and therefore covers both commensurate and incommensurate fluxes, where no finite magnetic unit cell exists. On connected uniformly gapped regions meeting the non-interacting plane $\lambda=0$, we prove a many-body gap-labeling theorem: the Hall conductivity appearing in the macroscopic Ohm's law is constant and quantized, satisfying $2\pi\sigma^{\mathrm H}\in\mathbb Z$. Thus the integer colors of the non-interacting Hofstadter butterfly persist as Hall-conductivity labels of interacting quantum Hall phases.
\end{abstract}

\section{Introduction}

The colored Hofstadter butterfly is the paradigmatic phase diagram of the integer quantum Hall effect for non-interacting lattice electrons with one-body Hamiltonian $\mathfrak{h}^b$ in a magnetic field $b$ at chemical potential $\mu$. The phases are the connected components of
\[
\{(b,\mu)\in\R^2\;|\; \mu\notin\sigma(\mathfrak h^b)\},
\]
which are open sets as the map $b\mapsto \sigma(\mathfrak{h}^b)$ is  continuous in the Hausdorff metric.
The color of such a component encodes an integer label $2\pi\sigma^{\mathrm H}$, which for commensurate magnetic fields $b\in2\pi\Q$ can be computed as the Chern number of a complex vector bundle. The physical meaning of 
$\sigma^\mathrm{H}$ is the Hall conductivity within that phase. The study of these and more general models for non-interacting electrons has led to deep and beautiful mathematics, drawing from spectral analysis, topology, non-commutative geometry, number theory, and dynamical systems theory.
A longstanding and in parts still open problem is whether this  labeled phase diagram retains a many-body meaning.
This is far from obvious   as  interactions destroy large parts of the structures basic to the analysis in the non-interacting case.
The question is whether the gapped regions nevertheless persist for interacting systems, and whether their integer labels still give the Hall conductivity appearing in a macroscopic Ohm's law.

We prove that the answer is  positive for a broad class of Hofstadter-like lattice fermion systems on the lattice $\Z^2$ with short-range interactions. More precisely, for every $(b_0,\mu_0)$ with $\mu_0\notin\sigma(\mathfrak h^{b_0})$, we construct an open neighborhood in the three-dimensional parameter space $(b,\mu,\lambda)$ of magnetic field, chemical potential, and interaction strength on which the interacting system has locally unique gapped ground states. On connected uniformly gapped regions meeting the non-interacting plane $\lambda=0$, the macroscopic Hall conductivity is constant and satisfies $2\pi\sigma^{\mathrm H}\in\Z$.

A central difficulty, and one of the technical novelties of the paper, is that the magnetic field is treated as a real parameter. The fractal structure of the Hofstadter butterfly stems from the interplay between two spatial scales and the relevant dimensionless parameter is the magnetic flux per unit cell in units of the flux quantum, namely $b/(2\pi)$.  If $b/(2\pi)=p/q\in\Q$, the system is periodic with respect to the sublattice $\Z\times q\Z\subset \Z^2$, while if $b/(2\pi)\notin\Q$, no finite magnetic unit cell exists. As a consequence, for  $b/(2\pi)\notin\Q$ no exact periodizations of the model on   finite tori exist. To overcome this difficulty, we work directly in the infinite-volume CAR algebra $\mA$ and obtain estimates that are uniform in the denominators of rational approximants, allowing us to pass to arbitrary fluxes.

\begin{figure}[ht]
\centering
\includegraphics[width=0.95\textwidth]{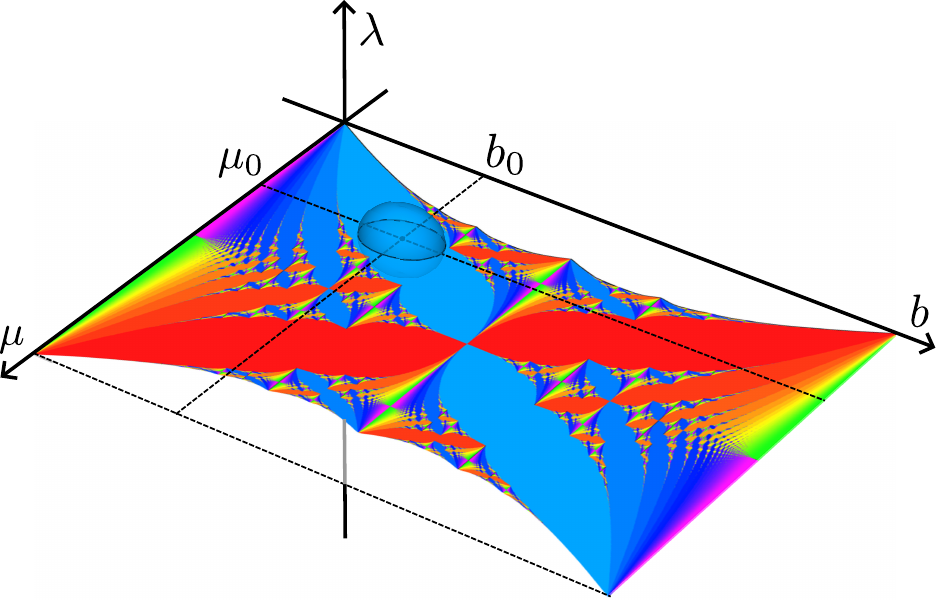}
\caption{The colored Hofstadter butterfly as a phase diagram for interacting systems. 
The colors in the $b$-$\mu$ plane indicate the integer gap labels of the spectral gaps of the non-interacting model, $\lambda=0$, cf.~\cite[Fig.~1]{osadchy2001hofstadter}. 
We prove that these gapped non-interacting regions persist for sufficiently weak interactions, $\lambda\neq 0$, and that the macroscopic Hall conductivity is constant throughout the corresponding interacting gapped phases. This is indicated by the local extension of the phase near $(b_0,\mu_0,0)$ into the interacting region~$\lambda\not=0$.
}\label{3Dbuttefly}
\end{figure}

More concretely, we study  interacting fermionic tight-binding models in a magnetic field $b\in\R$, implemented through Peierls phases; we call these systems interacting Hofstadter-like models. 
We first prove stability of the gapped phases of the non-interacting models under sufficiently weak interactions and then establish a gap-labeling theorem for integer quantum Hall phases in terms of Hall-conductivity gap labels. 
Our results give the colored Hofstadter butterfly a many-body interpretation: its gapped regions persist as interacting gapped phases, and its integer colors become macroscopic Hall-conductivity labels, see Figure~\ref{3Dbuttefly}.
In each gapped   phase connected to a Bloch--Landau level gap\footnote{We shall refer to the spectral islands obtained by the magnetic splitting of Bloch bands as
\emph{Bloch--Landau levels}. This terminology reflects the semiclassical picture: near a
non-degenerate extremum of a Bloch band and at small magnetic field, these islands are the
lattice analogues of continuum Landau levels, with effective mass determined by the curvature
of the underlying Bloch dispersion. Unlike continuum Landau levels, however, they need not be
flat; in the Hofstadter problem they are spectral islands separated by magnetic gaps.} of the non-interacting model, a \emph{macroscopic Ohm's law} holds with a \emph{quantized Hall conductivity} ($2\pi\sigma^{\mathrm H}\in\Z$) that is constant throughout the phase.  This extends known universality results for the Hall conductivity in several directions, most notably by treating $b$ explicitly as a continuous parameter and by identifying the gap label with the proportionality constant in Ohm's law for macroscopic currents, rather than merely with the associated linear-response coefficient.

Technically, our work builds on \cite{WMMMT24}, where we established a macroscopic Ohm's law for gapped lattice fermion systems in infinite volume and proved invariance of the Hall conductivity under locally generated automorphic equivalence. The present paper combines this with spectral stability of weakly interacting fermions \cite{de2019persistence,hastings2019stability,GMP2017} and with denominator-independent approximations of magnetic Hamiltonians. For fixed $b$, the interacting ground state is obtained from the non-interacting one by quasi-adiabatic continuation in $\lambda$. The dependence on $b$ is handled instead through magnetic perturbation theory and  rational approximants,   rather than by constructing an automorphic evolution in the direction of varying magnetic field.

Before turning to more technical details, let us recall that the non-interacting theory underlying the Hofstadter butterfly began 50 years ago and has led to deep developments in spectral theory, topology, and noncommutative geometry.
  The spectral study of discrete magnetic one-body Schr\"odinger operators began with Hofstadter's observation \cite{hofstadter1976energy} that the spectrum, plotted as a function of the magnetic field, displays a fractal structure. Shortly after the discovery of the integer quantum Hall effect \cite{klitzing1980new}, the TKNN formula identified the linear-response Hall conductivity at rational flux with the Chern number of a vector bundle over the magnetic Brillouin torus \cite{TKNN,avron1983homotopy,niu1985quantized}. The colors of the Hofstadter butterfly encode these integer Hall labels \cite{dana1985quantised,osadchy2001hofstadter}; this colored version and its interpretation as a phase diagram were emphasized in \cite{osadchy2001hofstadter}. For irrational flux and disordered systems, noncommutative geometry and gap-labelling methods were developed as a  replacement for Bloch theory \cite{bellissard1992gap,bellissard1994noncommutative}. 
  The spectral theory at incommensurate flux has also been studied by dynamical-systems methods, culminating in the resolution of the famous Ten Martini Problem in~\cite{avila2009ten}.

In the last decade there has been also substantial progress on interacting quantum Hall systems.
Integer quantization of the many-body Hall conductance has been proved by flux-threading and quasi-adiabatic methods under suitable finite-volume gap and uniqueness assumptions
\cite{hastings2015quantization,bachmann2018quantization}, and  infinite-volume results for weakly interacting periodic systems using renormalization group methods were obtained in \cite{GMP2017}. 
These works do not, however, establish the persistence of Hofstadter gapped regions under interactions with estimates uniform in the denominator $q$ of the rational flux $b/(2\pi)=p/q$, nor do they construct arbitrary-flux interacting Hofstadter phases carrying the Hall labels of the colored butterfly. This is the point addressed here.

This first brief overview has, of necessity, left out a number of important results; we will comment on the literature in a bit more detail at the end of the introduction. Before doing so, let us provide some further details on our result and our method.
In \cite{WMMMT24} we proved a macroscopic Ohm's law for general translation-invariant fermion systems on $\Z^d$ with short-range interactions and initially in a locally unique gapped ground state $\w_0$ using the NEASS approach:
The results of  \cite{BTW26adiabatic} (see also \cite{Teufel2020,HenheikTeufel2022}) imply that 
when adiabatically switching on a constant external electric field of strength $\epsi$, the system is driven into 
 a non-equilibrium almost stationary state (NEASS) denoted by $\w_\epsi$.
Macroscopic Ohm's law\footnote{In high-precision quantum Hall experiments the measured quantity is not the linear-response coefficient, that is the conductance in the limit of vanishing current. Instead, large Hall currents and voltages are required to keep relative errors in their quotient small. 
It is precisely Ohm's law that connects the measured conductance with the theoretical linear-response coefficient.
Mathematically,   the difference between Ohm's law and linear response is $\mathcal{O}(\epsi^\infty)$ versus  $o(\epsi)$ in \eqref{OLintro}.}
is the statement that the current density in this NEASS   is nearly linear in the applied field strength,
\begin{equation}\label{OLintro}
\overline{\w_\epsi}(J^\mathrm{H}) = \epsi\, \sigma^\mathrm{H} + \mathcal{O}(\epsi^\infty)\,.
\end{equation}
Here the notation $\overline{\w_\epsi}(J^\mathrm{H})$ indicates that an average of an extensive observable is taken,
\[
\overline{\w_\epsi}(J^\mathrm{H}) := \lim_{\Lambda\nearrow\Z^2} \,\frac{\w_\epsi(J^\mathrm{H}|_{\Lambda})}{|\Lambda|}\,,
\]
where $J^\mathrm{H}$ is the current operator perpendicular to the applied electric field. This is why we speak of `macroscopic' currents and `macroscopic'  conductivity.
This macroscopic Hall conductivity $\sigma^\mathrm{H}$ can be computed from $\w_0$ as follows, 
\begin{equation}\label{CHintro}
\sigma^\mathrm{H} =\overline{\w_0} \left(\i \left[X_2\OD ,X_1\OD \right]\right) \,. 
\end{equation}
Details on the definitions of the off-diagonal position operators $X_i\OD$ are given below. In the present paper we show that these results apply to   interacting Hofstadter-like models with Hamiltonians of the form
\[
H_{(b,\mu,\lambda)} := \mathrm{d}\Gamma(\mathfrak{h}^b - \mu \, \mathbf{1}) + \lambda V^b\,,
\]
where $\mathfrak{h}^b$ is a short-range periodic one-body operator with magnetic field $b$ implemented through Peierls phases, $\mu\in\R$ is the chemical potential, and $\lambda\in\R$ controls the strength of the many-body interactions. Our main result   states that whenever $\mu_0 \notin\sigma(\mathfrak{h}^{b_0})$, then there exist $\delta, \lambda_0>0$ such that for all $(b,\mu,\lambda)$ in  the cylinder   $B_\delta(b_0,\mu_0)\times (-\lambda_0,\lambda_0)$
around $(b_0,\mu_0)$ in three-dimensional parameter space
\begin{itemize}
    \item[(i)] $H_{(b,\mu,\lambda)}$ has a locally unique gapped ground state $\w_{(b,\mu,\lambda)}$ and thus, according to \cite{WMMMT24},   Ohm's law~\eqref{OLintro} holds and the Hall conductivity $\sigma^\mathrm{H}_{(b,\mu,\lambda)}$ is given by \eqref{CHintro};
    \item[(ii)]
    the Hall conductivity is quantized and constant, i.e.\ $\sigma^\mathrm{H}_{(b,\mu,\lambda)}\equiv \sigma^\mathrm{H} \in \frac{1}{2\pi}\Z$.
\end{itemize}
The situation is depicted in Figure~\ref{3Dbuttefly} for the weakly-interacting Hofstadter model. Indeed, (ii) holds not only on $B_\delta(b_0,\mu_0)\times (-\lambda_0,\lambda_0)$, but  on a large class of connected sets in parameter space on which the spectral gap does not close and that intersect the $\lambda=0$ plane.

We briefly place our results in the broader context of quantum Hall theory. 
More than forty years after its discovery, the quantum Hall effect, both integer \cite{vonKlitzing1986quantized} and fractional \cite{stormer1999fractional}, continues to pose stimulating theoretical challenges and remains an active area of research. The integer quantum Hall effect, with which the present work is concerned, occupies a central place in modern condensed matter theory: it was the first experimentally observed manifestation of topologically quantized transport, and 
the quantum Hall resistance provides a primary realization of the ohm ($\Omega$)  with exceptional accuracy \cite{klitzing1980new, he2022accurate}.
The quantum Hall effect is also one of the most striking phenomena in which genuine quantum behavior becomes visible  at macroscopic scales. 
Several complementary theoretical frameworks have been developed to explain this phenomenon from first principles: the Chern-number formula of Thouless--Kohmoto--Nightingale--den Nijs and its extensions, Laughlin's flux-insertion argument, edge-state and Landauer--B\"uttiker approaches, Fr{\"o}hlich's gauge-principle approach based on current algebra and  effective field theory, and Bellissard's noncommutative-geometric/Fredholm-index formulation
\cite{TKNN,avron1983homotopy,niu1985quantized,halperin1982quantized,laughlin1981quantized,avron1985quantization,buttiker1988absence,frohlich1993gauge,bellissard1994noncommutative}.
Within mathematical physics, much of the rigorous literature focuses on the linear-response coefficient, corresponding to the conductance in the limit of vanishing electric field. In high-precision quantum Hall experiments, however, one measures macroscopic currents and voltages. The bridge between the two regimes is Ohm's law.

We close with two further comments.
\begin{itemize}
    \item
    The results on microscopic Ohm's laws proved in \cite{klein1990power,bachmann2021exactness,TW2025}  extend to the class of models considered here by arguments similar to those developed in the present paper. 
Indeed, for \cite{TW2025} this is evident and the methods developed here should also allow to  verify   the kind of uniform gap assumption along a pumping cycle required in \cite{klein1990power,bachmann2021exactness}.
\item 
Tracking the ground-state evolution along changes of the magnetic field strength does not, in general, fit the standard setup of automorphic equivalence \cite{BMNS12}, not even in the non-interacting regime (cf.~\cite{cornean2021beyond}). 
The obstruction is that the charge density of a gapped ground state can vary with $b$, whereas locally generated automorphisms generated by bounded gauge-invariant interactions preserve the charge density. 
In Remark~\ref{autoRemark} we briefly discuss a possible enlargement of the notion of automorphic equivalence using spatially unbounded generators. Existence of dynamics for such generators has recently been established in~\cite{TWW26}.
\end{itemize}

\paragraph{Structure of the paper.}

In Section~\ref{sec:MR}, we present the class of models analyzed in this paper, namely interacting Hofstadter-like models, and we formulate precisely our main results: Theorem~\ref{prop:Hofstadter}, concerning the existence of interacting gapped ground states, whose proof is presented in Section~\ref{sec:Hofstadter}; and Theorem~\ref{cor:universality}, showing the constancy (and \text{a posteriori} quantization) of the Hall conductivity within these gapped phases, whose proof is presented in Section~\ref{sec:universality}. The appendices are devoted to recalling the mathematical framework to describe quantum systems of interacting fermions on a lattice by means of an algebra of (quasi-)local observables (Appendix~\ref{sec:basic}), and to collect some technical tools which are used in the proofs of the main results (Appendix~\ref{appendix_technical}).

\paragraph{Acknowledgments.}
This work was funded by the \foreignlanguage{ngerman}{Deutsche Forschungsgemeinschaft} (DFG, German Research Foundation) – 470903074; 465199066. 
G.~M.\ gratefully acknowledges financial support from the European Research Council through the ERC CoG UniCoSM, grant agreement n.724939. 
T.~M.~was supported by JSPS KAKENHI Grant Numbers 20KK0304 and 23H01086. 
D.~M.~gratefully acknowledges financial support from Sapienza Universit\`{a} di Roma within Progetto di Ricerca di Ateneo 2023, 2024 and~2025. 
S.~T.~thanks the Mathematics Department at Sapienza Universit\`{a} di Roma for their hospitality and financial support during a two-month research visit.

\section{Model and results} \label{sec:MR}

In order to keep the present section concise, we have moved some basic notions, together with several more specialized technical definitions, to Appendix~\ref{sec:basic}. These will be used freely throughout the section.

\subsection{Interacting Hofstadter-like Hamiltonians} \label{sec:models}

We consider periodic, infinitely extended systems of interacting fermions in two spatial dimensions, subject to a transverse magnetic field of strength $b$.
A key feature of our approach is that the notion of periodicity we impose is compatible with arbitrary magnetic field strengths $b\in\mathbb{R}$.
The Hamiltonians are defined as the sum of a non-interacting part, to which the magnetic field is coupled through Peierls phases, and an interacting part,   which is compatible with the magnetic field; for example  standard density-density interactions.

To define the non-interacting   Hamiltonians $\mathfrak{h}^b$ acting on the one-body space $ \ell^2(\Z^2,\C^\fd)$, let first 
  $h:\Z^2\to \mL(\C^\fd)$ satisfy  $h(-z) =   h(z)^*$ for all $z\in\Z^2$ and have exponential decay, i.e., there exist  $c,C>0$ such that
\[
\forall  z\in\Z^2\,:\; \|h(z)\|\leq C\,\e^{-c\|z\|}\,.
\]
Throughout, for $z\in\Z^2$ the norm $\|z\|$ denotes the $\ell^\infty$-norm and for linear maps the operator norm.
For any $b\in\R$, the   operator $\mathfrak{h}^b$ is obtained by multiplying the translation-invariant hopping terms by a Peierls phase:
 \[
\mathfrak{h}^b(x,y)_{ij} := \mathrm{e}^{ \mathrm{i}\frac{x_2+y_2}{2} b (x_1-y_1)} h(x-y)_{ij}\,.
 \]
Here we define the operator $\mathfrak{h}^b$ in terms of its summation kernel given by 
\begin{equation}\label{kernelDef}
\mathfrak{h}^b(x,y)_{ij}  = \langle e_x\otimes v_i, \mathfrak{h}^b\,e_y\otimes v_j\rangle
\end{equation}
with respect to the canonical orthonormal bases $\{e_x \,| \,x\in\Z^2\}$ of $\ell^2(\Z^2)$ and $\{v_i\, |\, i\in \{1,\ldots,\fd\}\}$ of $\C^\fd$. Notice how, if one sets $\fd=1$ and 
\[ h(z) \equiv h_{\mathrm{Hofstadter}}(z) = \begin{cases} 
1 & \text{if } |z| = 1, \\
0 & \text{otherwise}, 
\end{cases} \]
one recovers as $\mathfrak{h}^b \equiv \mathfrak{h}_{\mathrm{Hofstadter}}^b$ the usual Hofstadter Hamiltonian, namely the discrete magnetic Laplacian on the lattice $\Z^2$ \cite{hofstadter1976energy}.
 
The (dual) magnetic translations are defined for any $\gamma\in\Z^2$ as
\[
t^b_\gamma:\ell^2(\Z^2,\C^\fd)\to \ell^2(\Z^2,\C^\fd)\,,\quad (t^b_\gamma f)(x) := \e^{-\i b x_1\gamma_2}f(x+\gamma)\,,
\]
with inverse
\[
\big(t^b_\gamma\big)^{-1}:\ell^2(\Z^2,\C^\fd)\to \ell^2(\Z^2,\C^\fd)\,,\quad  \big[\big(t^b_\gamma\big)^{-1} f\big](x) := \e^{\i b (x_1-\gamma_1)\gamma_2}f(x-\gamma)\,.
\]
Then, $\mathfrak{h}^b$ is $t^b$-invariant, i.e.
\begin{equation} \label{eqn:hb_tb_inv}
\big(t_\gamma^b\big)^{-1}\mathfrak{h}^b  t_\gamma^b = \mathfrak{h}^b \quad \text{for all } \gamma\in\Z^2\,.
\end{equation}
Magnetic translations in general yield merely a projective representation of the abelian group $\Z^2$, as
\[ t_\gamma^b \, t_{\gamma'}^b = \e^{-\i b \gamma_1 \gamma_2'} t_{\gamma+\gamma'}^b = \e^{-\i b(\gamma_1 \gamma_2'-\gamma_1'\gamma_2)}  t_{\gamma'}^b \, t_{\gamma}^b\,, \quad \gamma, \gamma' \in \Z^2.  \]
Only in the case where $b = 2\pi \tfrac{p}{q} \in 2 \pi \mathbb{Q}$ is commensurate to $2\pi$, magnetic translations provide a true unitary representation of the sublattice $\Z^2_q := \Z \times (q\Z) \subset \Z^2$, and operators which commute with them are amenable to being studied via the Bloch--Floquet representation (compare Appendix~\ref{app:BF}). In view of these remarks, later at times we will have to distinguish between ``commensurate'' and ``incommensurate'' magnetic fields $b$ when studying spectral properties of magnetic Hamiltonians. At any rate, the results we present   hold for any $b \in \R$ regardless of this commensurability distinction.

It will be convenient to write Hamiltonians and other operators (one-body  and many-body) as sums of local terms.
For $\mathfrak{h}^b$, and similarly for  other one-body operators, define $R_x :=\{(y,z)\in \Z^2\times\Z^2\,|\, x_j= \lfloor \frac{y_j+z_j+1}{2} \rfloor \,,\;j\in\{1,2\}\}$ and $\mathfrak{h}^b_x$ to be the operator with kernel
\[
\mathfrak{h}^b_x(y,z) := \delta_{(y,z)\in R_x} \mathfrak{h}^b(y,z)\,.
\]
In short, every ``hopping from $z$ to $y$'' is associated with the location $x$ at the middle between $z$ and $y$.
Then $\mathfrak{h}^b =\sum_{x\in\Z^2}\mathfrak{h}^b_x$ and  $t^b$-invariance of $\mathfrak{h}^b$ implies that 
\[
\mathfrak{h}^b_x = \big(t^b_x\big)^{-1} \mathfrak{h}^b_0\, t^b_x  \,.
\]

The second quantization of $\mathfrak{h}^{b}$ at chemical potential $\mu\in\R$ is defined by
\[
H_{(b,\mu,0)} := \mathrm{d}\Gamma(\mathfrak{h}^b - \mu \, \mathbf{1}) := \sum_{\substack{x,y\in\Z^2}} \sum_{i,j=1}^\fd \, a_{x,i}^* \,\mathfrak{h}^b(x,y)_{ij}\,a_{y,j} - \mu\sum_{x\in\Z^2 } \sum_{i=1}^\fd n_{x,i}\,.
\]
The operator $H_{(b,\mu,0)}$   is expressed in terms of fermionic creation and annihilation operators $a_{x,i}^*$ and $a_{x,i}$, respectively, labeled by lattice sites $x \in \Z^2$ and internal indices $i \in \{1,\ldots,\fd\}$, and acting on the fermionic Fock space $\mathcal{F}(\Z^2;\C^\fd)$ over $\ell^2(\Z^2,\C^\fd)$; the second quantization map of a 1-body  self-adjoint  operator is denoted, as customary, by $\mathrm{d}\Gamma(\cdot)$, and $n_{x,i} := a_{x,i}^* \, a_{x,i}$ is the local number operator. We will also denote by $n_x := \sum_{i=1}^{\fd} n_{x,i}$ the fermion number operator over site $x$. 
We do not consider $H_{(b,\mu,0)}$ as an unbounded operator on fermionic Fock space, but as an interaction   on the quasi-local CAR-algebra $\mA$  generated by $\{a_{x,j}^*, a_{x,j} \,|\, x \in \Z^2, \, j \in \{1, \dots , \fd\}\}$, see Appendix~\ref{sec:basic} for details. Explicitly, as an interaction  it is the map 
$H_{(b,\mu,0)}:P_0(\Z^2) \to \mA, \, Z\mapsto     H_{(b,\mu,0)}(Z) $, with
\begin{align*}
H_{(b,\mu,0)}(Z) \coloneq 
    \begin{cases}
    \sum_{i,j=1}^\fd \, a_{x,i}^* \,\mathfrak{h}^b(x,x)_{ij}\,a_{x,j}  - \mu \sum_{i=1}^\fd n_{x,i} & \text{if}\,\, Z=\{x\},  \\
        \sum_{i,j=1}^\fd \, a_{x,i}^* \,\mathfrak{h}^b(x,y)_{ij}\,a_{y,j} + a_{y,i}^* \,\mathfrak{h}^b(y,x)_{ij}\,a_{x,j} & \text{if}\,\, Z=\{x,y\}, \, x\not=y \\
       0 & \text{ otherwise.}
    \end{cases}
\end{align*}
The Hamiltonian  $H_{(b,\mu,0)}$ is then an element of $B_{\exp}$, the set of short-range interactions on the algebra  $\mA$.
General properties of local and quasi-local operators on the CAR-algebra are reviewed in Appendix~\ref{sec:basic}, cf.\ also \cite[Section 2.1]{WMMMT24}. 
 
The magnetic translations act on the CAR-algebra as automorphisms $T^b_\gamma \in \aut(\mathcal{A})$, $\gamma \in \Z^2$, and are  fully defined by their action on annihilation operators as
\begin{equation}\label{magtrans}
a_{y,j}\mapsto T^b_\gamma(a_{y,j}) :=  \mathrm{e}^{-\mathrm{i}b y_1 \gamma_2} a_{y+\gamma,j} \,,\quad y\in\Z^2, \, j\in \{1,\dots\fd\}.
\end{equation}
By construction, $H_{(b,\mu,0)}$ can be also written as a 0-chain  on $\mathcal{A}$, i.e.\  as a sum of local terms which are translates of a single term at the origin:
\[
H_{(b,\mu,0)} = \sum_{x\in\Z^2} T^b_x ( \mathrm{d}\Gamma(\mathfrak{h}^b_0) - \mu\, n_0) =: \sum_{x\in\Z^2} H_{(b,\mu,0),x} \,.
\]
Finally we also add a  short-range interaction  with coupling parameter $\lambda\in\R$: For each $b\in \R$ let 
$v^b \in\mA^N_{\exp}$ be self-adjoint ,   $T^b$-compatible, i.e.\ $T^b_\gamma T^b_{\gamma'} v^b = T^b_{\gamma+ \gamma'} v^b $ for all $\gamma,\gamma'\in\Z^2$, and such that the map   $\R\to \mA_{\exp}\subset\mA_\infty$, $b\mapsto v^b$, is continuous  with respect to all decay norms $\|\cdot\|_\nu$.  
For the detailed definitions of  $\mathcal{A}_{\rm exp}^{N}$, $\mathcal{A}_{\rm exp}$, $\mathcal{A}_{\infty}$, and $\|\cdot\|_{\nu}$, we refer the reader to Appendix~\ref{sec:basic}.
 We define $V^b$ to be the $T^b$-periodic interaction induced by translating $v^b$ across the lattice as defined in Proposition \ref{prop: interaction associated to observable},
and
\[
H_{(b,\mu,\lambda)} := H_{(b,\mu,0)} + \lambda V^b\,.
\] 
The standard example of a $T^b$-compatible interaction term is an interaction that only depends on the local particle densities, e.g.\ a two-body interaction of the form
\[
v^b = v = \tfrac{1}{2}\sum_{x\in\Z^2} w(x) n_0 n_x\,,
\]
where $w:\R^2\to\R$ is an exponentially decaying even function.
Note that from now on we freely use the identification of  $0$-chains and interactions given in Proposition~\ref{prop: interaction associated to observable}.

\subsection{Main results}

We view the above $H_{(b,\mu,\lambda)}$ as defining a three-parameter family of models with $(b,\mu,\lambda)\in \R^3$. This operator gives rise to a densely defined derivation $\mL_{H_{(b,\mu,\lambda)}}$ on the CAR-algebra $\mA$, and one says that a state $\w$ on $\mA$ is a \textbf{locally unique gapped ground state} with gap $g>0$ for $H_{(b,\mu,\lambda)}$ iff it holds that
\[
\forall A \in \mA_{0}\;:\qquad
\w(A^* \mL_{H_{(b,\mu,\lambda)}}A)
\geq g\left(\w(A^*A)- |\w(A)|^2
\right) \,, 
\]
 where  $\mA_0$  is  the local algebra obtained as the union of all compactly-supported expressions in the creation and annihilation operators.

We are interested in the gapped phases of $H_{(b,\mu,\lambda)}$, which we define for the purpose of this paper as follows.

\medskip

\begin{definition}
A {\bf uniformly gapped  region} for the family $(H_{(b,\mu,\lambda)})_{(b,\mu,\lambda)\in\R^3}$ is an open connected non-empty set $\mP \subset\R^3$ together with a family of states $(\rho_{p})_{{p}\in \mP}$ such that:
\begin{enumerate}
    \item[(i)]
There exists $g>0$ such that for all $p \in \mP $ the state 
$\rho_{p}$ is a locally unique gapped ground state for $H_{p}$ with gap $g$.
\item[(ii)]
For any $A\in \mA$    the map 
\[
\mP \to \C\,,\quad  
(b,\mu,\lambda)\mapsto \rho_{(b,\mu,\lambda)} (A)
\]
is continuous. 
\end{enumerate}
\end{definition}

\medskip

Still at an informal level, our results are as follows.
\begin{enumerate}
 \item 
 We first show in Theorem~\ref{prop:Hofstadter} that every point of a non-interacting gapped region has a neighborhood that extends to a uniformly gapped region of the interacting model in the following sense: For each $\mu_0\notin \sigma(\mathfrak{h}^{b_0})$, we find $\delta,\lambda_0>0$ such that $\mathcal{P}_{\delta,\lambda_0} := B_\delta(b_0,\mu_0 )\times (-\lambda_0,\lambda_0)$, together
 with an explicitly constructed family of states $(\w_{(b,\mu,\lambda)})_{(b,\mu,\lambda)\in \mathcal{P}_{\delta,\lambda_0}}$,
 is a uniformly gapped region for $H_{(b,\mu,\lambda)}$.   The proof is based on results on the stability of the spectral gap of free fermion systems \cite{de2019persistence, hastings2019stability, GMP2017}, but the extension to infinite systems with possibly incommensurate $b$ requires additional arguments in view of the above-mentioned lack of periodicity.
 \item Having established the existence of gapped phases for our weakly interacting generalized Hofstadter Hamiltonians, our main new result is a ``gap labeling'' theorem: we show that the Hall conductivity, defined by the many-body version of the double commutator formula derived in \cite{WMMMT24}, is constant and quantized, with $2\pi\sigma^{\mathrm H}\in\Z$,  in the uniformly gapped regions of $H_{(b,\mu,\lambda)}$ around $\lambda=0$. This is the content of Theorem~\ref{cor:universality}, which gives rise to an extended many-body phase diagram ``colored'' like in Figure~\ref{3Dbuttefly}.
\end{enumerate}

 A key technical ingredient in our proof is the construction of a locally generated automorphism $\alpha_{(b,(0,\lambda))}$ of $\mA$ that  maps    the unique gapped ground state of the non-interacting system $H_{(b,\mu,0)}$   to the gapped ground state constructed below for the weakly interacting system  $H_{(b,\mu,\lambda)}$. This automorphism is in principle well known: in the case of the standard adiabatic theorem of quantum mechanics formulated for self-adjoint operators on Hilbert spaces it is given by the unitary parallel transport map in the vector bundle defined by the ground state projections $\lambda\mapsto P_\lambda$. In the many-body setting one speaks of automorphic equivalence of gapped ground states, see  \cite{BMNS12} for finite systems and \cite{MO20,BTW26automorphic} for infinitely extended systems. In \cite{MO20,BTW26automorphic}, however, one starts with a sufficiently regular family of gapped ground states and shows that they are connected by locally generated automorphisms. Here we will turn the logic around: we define the states through parallel transport and show that they are indeed gapped ground states.

To this end let us introduce the generator of this so-called quasi-adiabatic evolution   to be the $B_\infty$-interaction $K^{(b,\lambda),g}$ with $K^{(b,\lambda),g}_0$ given by 
\begin{equation}\label{KDef}
K^{(b,\lambda),g}_0  :=  -\int_\R \D s\, W_g(s)
\,\e^{\i s\mL_{H_{(b,\mu,\lambda)}}} v^b
\,,
\end{equation}
where $W_g$ is a rapidly decaying function depending on a gap parameter $g>0$, defined as in Appendix~\ref{sec:offdiag}; note that $g$ will be fixed in the proofs. The interaction
$K^{(b,\lambda),g}$ generates the strongly continuous cocycle of automorphisms
$(\lambda_0,\lambda)\mapsto \alpha_{(b,(\lambda_0, \lambda))}$, the quasi-adiabatic evolution,  satisfying
\[
\partial_\lambda \,\alpha_{(b,(\lambda_0, \lambda))} = \alpha_{(b,(\lambda_0, \lambda))} \circ \i \mL_{K^{(b,\lambda),g}}\,.
\]
Note that $K^{(b,\lambda),g}$ and thus also $\alpha_{(b,(\lambda_0, \lambda))}$ do not depend on $\mu$, since $v^b$ is by assumption gauge-invariant and thus   commutes with the number operator. For better readability we also not make explicit the $g$-dependence of $\alpha_{(b,(\lambda_0, \lambda))}$.

Assume that $\mu\notin\sigma(\mathfrak{h}^{b})$ and let $P_{(-\infty,\mu)}^{\mathfrak{h}^b} := \chi_{(-\infty,\mu)}(\mathfrak{h}^b)$ denote the corresponding Fermi projection below the spectral gap around $\mu$. 
Then the quasi-free state $\w_{(b,\mu,0)}$, characterized uniquely by its two-point function 
$\w_{(b,\mu,0)}(a^*_{x,i}a_{y,j}) = P_{(-\infty,\mu)}^{\mathfrak{h}^b}(x,y)_{ij}$ according to Wick's theorem, is the  unique infinite-volume ground state of  $H_{(b,\mu,0)}$. Here  $P_{(-\infty,\mu)}^{\mathfrak{h}^b}(x,y)_{ij}$ are the matrix elements of $P_{(-\infty,\mu)}^{\mathfrak{h}^b}$ with respect to the canonical basis of $\ell^2(\Z^2,\C^\fd)$ defined as in \eqref{kernelDef}. 
 Given $g>0$, we define for $(b,\mu,\lambda)\in \R^3$  the state 
\begin{equation}\label{omegadef}
\w_{(b,\mu,\lambda)} \coloneq \w_{(b,\mu,0)} \circ \alpha_{(b,(0, \lambda))}\,.
\end{equation}
Notice that $\w_{(b,\mu,0)}$ and $\w_{(b,\mu,\lambda)}$ do not depend on $\mu$ as long as $\mu$ lies within the same gap of the one-particle Hamiltonian $\mathfrak{h}^{b}$. Still we keep $\mu$ as a label of the state to indicate that the gapped phase in which $\w_{(b,\mu,\lambda)} $ lies does depend on $\mu$.

We can now formulate our first main result. 

\begin{theorem}\label{prop:Hofstadter}
    Let $(b_0,\mu_0)\in\R^2$ be such that $\mu_0\notin\sigma(\mathfrak{h}^{b_0})$. Then there exist $\delta,\lambda_0,g_*>0$   such that 
    \[\mathcal{P}_{\delta,\lambda_0} := B_{\delta}(b_0,\mu_0)\times (-\lambda_0,\lambda_0)\] 
    together with the family of states 
    \[(\w_{(b,\mu,\lambda)})_{(b,\mu,\lambda) \in \mathcal{P}_{\delta,\lambda_0}}\,,\] 
    defined as in \eqref{omegadef} with respect to $g=g_*$,
    form a uniformly gapped region with gap~$g_*$ for $(H_{(b,\mu,\lambda)})_{(b,\mu,\lambda)\in \R^3}$. Furthermore, for all $(b,\mu,\lambda)\in \mathcal{P}_{\delta,\lambda_0}$ the state $\w_{(b,\mu,\lambda)}$ is $T^{b}$-invariant and   $\w_{(b,\mu,\lambda)} = \w_{(b,\mu_0,\lambda)}$.
\end{theorem}   
\noindent The proof of Theorem~\ref{prop:Hofstadter} is presented in Section~\ref{sec:Hofstadter}.

\smallskip

The second theorem concerns the Hall current response of the states~\eqref{omegadef} to external electric fields. It has been shown in~\cite{WMMMT24} that, when the system modeled by $H_{(b,\mu,\lambda)}$ is initially in the gapped ground state $\w_{(b,\mu,\lambda)}$ and  is then perturbed by an electric field of intensity $\eps$ in the direction of $x_1$, then the transverse current in the direction of $x_2$ flowing in the corresponding non-equilibrium almost-stationary state (NEASS) $\w_{(b,\mu,\lambda)}^{\eps} := \w_{(b,\mu,\lambda)} \circ \beta^\epsi= \w_{(b,\mu,\lambda)} \circ \e^{\i\mL_{S^\epsi}}$ is  
\[ j_2^{\epsi} = \epsi\, \sigma^\mathrm{H}_{(b,\mu,\lambda)} + \mathcal{O}(\epsi^\infty)\,. \]
Here 
\begin{equation} \label{sigmaHdef}
\sigma^\mathrm{H}_{(b,\mu,\lambda)} := \overline{\w}_{(b,\mu,\lambda)}\left(\i \left[X_2\OD[(b,\mu,\lambda)],X_1\OD[(b,\mu,\lambda)]\right]\right) 
\end{equation}
is the Hall conductivity and $\overline{\w}_{(b,\mu,\lambda)}(\cdot)$ denotes the expectation-per-unit-volume in the state $\w_{(b,\mu,\lambda)}$ (compare \eqref{eq:omegaPUV}). The off-diagonal part of the position operator $X_j := \sum_{x \in \Z^2} x_j \, n_{x}$, $j \in \{1, 2\}$, with respect to the state $\w_{(b,\mu,\lambda)}$ is
\begin{equation} \label{ODinformal}
X_j\OD[(b,\mu,\lambda)] := \i \, \int_{\R}\mathrm{d}s \,W_g(s) \, \e^{\i s\mL_{H_{(b,\mu,\lambda)}}}\,  \mL_{X_j} H_{(b,\mu,\lambda)}\,.
\end{equation}
See Appendix~\ref{sec:offdiag} for details.
For details on how to construct the NEASS $\w_{(b,\mu,\lambda)}^{\eps}$ and under what conditions it provides the full response of a gapped system to a perturbation, see \cite{Teufel2020,HT20b,BTW26adiabatic}.

\begin{theorem}\label{cor:universality}
    Let $M_0 \subset \R^2$ be an open connected set such that 
 $\mu\notin \sigma(\mathfrak{h}^b)$ for all $(b,\mu)\in M_0$, i.e.\ part of a gapped phase of the non-interacting model. Assume that  $(H_{(b,\mu,\lambda)})$ admits a uniformly gapped region with gap $g>0$ of the form $\left( \mathcal{P}_{M_0}, \big(\w_{(b,\mu,\lambda)}\big)_{(b,\mu,\lambda)\in\mathcal{P}_{M_0}} \right)$, where  
$\mathcal{P}_{M_0}\subset M_0\times \R$ is open, contains $M_0\times\{0\}$, and  for each $(b,\mu,\lambda) \in \mathcal{P}_{M_0}$ contains also the line segment $[(b,\mu,0),(b,\mu,\lambda)]$, and where $\big(\w_{(b,\mu,\lambda)}\big)_{(b,\mu,\lambda)\in \mP_{M_0}}$ is the family of states defined in \eqref{omegadef} with respect to $g$. 
    Then the map
    \[
        \sigma^{\mathrm{H}} \colon \mathcal{P}_{M_0} \to \R, \quad (b,\mu,\lambda) \mapsto \sigma^\mathrm{H}_{(b,\mu,\lambda)} 
    \]
    is constant and equal to the integer-quantized value of the label $\sigma^\mathrm{H}_{M_0}$ of the non-interacting phase $M_0$, i.e.
    \[
    \sigma^\mathrm{H}_{(b,\mu,\lambda)}\equiv \sigma^\mathrm{H}_{M_0}\in \tfrac{1}{2\pi}\,\Z \,.
    \]
\end{theorem}
\noindent The proof of Theorem~\ref{cor:universality} is presented in Section~\ref{sec:universality}.

Note that Theorem~\ref{cor:universality} covers not only weak interactions, that is small $|\lambda|$, but all interactions 
for which $\mu $ effectively still lies in a gap between different Bloch--Landau levels. This is exactly the physical range for which integer quantization of Hall conductance is expected.

\begin{remark}\label{autoRemark}
For fixed magnetic field $b$, the states $\w_{(b,\mu,\lambda)}$ obtained by varying $\lambda$ are connected by the quasi-adiabatic cocycle generated by the bounded local interaction $K^{(b,\lambda),g}$ in~\eqref{KDef}. 
The situation is different when one varies the magnetic field. Formally, for a differentiable path $b\mapsto H_{(b,\mu,\lambda)}$, quasi-adiabatic perturbation theory would suggest a generator of the form
\[
\partial_b \,\w_{(b,\mu,\lambda)}
=
\w_{(b,\mu,\lambda)}\circ i\mL_{\widetilde K^{(b,\lambda),g}},
\]
with
\begin{equation}\label{tildeK}
\widetilde K^{(b,\lambda),g}_x
=
-\int_\R \D s\, W_g(s)\,
\e^{i s\mL_{H_{(b,\mu,\lambda)}}}
\partial_b H_{(b,\mu,\lambda),x}.
\end{equation}
In the magnetic models considered here, however, the local derivative
$\partial_b H_{(b,\mu,\lambda),x}$ grows linearly in the spatial coordinate, for instance like $|x_2|$ in the gauge used above. Thus \eqref{tildeK} does not define a bounded interaction in the usual sense.
There is a conceptual obstruction: within   phases  labeled by a non-zero Hall conductivity,  the charge density
$\overline\w_{(b,\mu,\lambda)}(N)$ changes along a path in $b$, and the corresponding states cannot be connected by a locally generated automorphism produced by bounded gauge-invariant interactions, since such automorphisms preserve the charge density per unit volume.

One can envisage different  approaches  to retain a notion of automorphic equivalence for magnetic gapped phases.
One could try  to stack $\w_{(b_1,\mu,\lambda)}$ and $\w_{ (b_2,\mu,\lambda)}$  with short-range entangled states and then deform the enlarged states using bounded gauge-invariant generators on the enlarged system along a path that departs from the family of states $\w_{(b,\mu,\lambda)}$. Alternatively one could look for bounded but not gauge-invariant generators. Both approaches have, to our knowledge, not been explored yet for magnetic phases. 

A third approach is to enlarge the class of   automorphisms  that facilitate automorphic equivalence   within  gapped ground state phases    by including  those generated by interactions with linear growth in space such as \eqref{tildeK}. In~\cite{TWW26} it is shown that such spatially growing interactions  generate automorphisms of $\mA$, providing a possible framework for automorphic equivalence within magnetic gapped phases. We will explore this approach in a future work.
\end{remark}

\section{Proof of Theorem \ref{prop:Hofstadter}} \label{sec:Hofstadter}
 Before we begin with the proof of Theorem \ref{prop:Hofstadter} let us first compare the present section with the proof of Proposition 3.4 in  \cite{WMMMT24}. There  it was shown that for all $(b_0,\mu_0) \in \R^2$ with $\mu_0 \notin \sigma(\mathfrak{h}^{b_0})$ there exists $\delta>0$ such that, for every $(b,\mu,\lambda)\in B_{\delta}(b_0,\mu_0,0)$, the Hofstadter Hamiltonian $H^{\mathrm{Hofstadter}}_{(b,\mu,\lambda)} = \mathrm{d} \Gamma \big(\mathfrak{h}_{\mathrm{Hofstadter}}^{b}-\mu\mathbf{1}\big)+\lambda V^b$ has a locally unique gapped ground state. We extend this result to Hofstadter-like models of the type presented in Section~\ref{sec:models}, and complement the argument by proving two crucial ingredients needed to show that the Hall conductivity remains constant within gapped phases of these models: namely, that  the explicitly constructed state  $\w_{(b,\mu,\lambda)}$ as defined in \eqref{omegadef} provides such a ground state, and that the map $B_{\delta}(b_0,\mu_0,0) \to S(\mA) , \, (b,\mu,\lambda)\mapsto \w_{(b,\mu,\lambda)}$ is continuous in the weak* topology,
where $S(\mathcal{A})$ denotes the set of all states on $\mathcal{A}$.

Concerning this last point, as was observed after \eqref{omegadef}, the states $\w_{(b,\mu,\lambda)}$ are constant along lines of fixed $b$ and $\lambda$ when $\mu \notin \sigma(\mathfrak{h}^b)$ (i.e.\ they depend on $\mu$ in a constant way, and {\it a fortiori} continuously). Moreover, continuity in $\lambda$ of $\w_{(b,\mu,\lambda)}$ follows at once from strong continuity of the family of automorphisms $\lambda \mapsto \alpha_{(b,(0,\lambda))}$, compare Section~\ref{continuity_in_lambda}. Furthermore, the states $\w_{(b,\mu,\lambda)}$ are $T^b$-invariant by construction:  $\w_{(b,\mu,0)}$ is $T^b$-invariant and  $\alpha_{(b,(0,\lambda))}$ is generated by the $T^b$-invariant interaction $K^{(b,\lambda),g}$.

Proving that the states $ \w_{(b,\mu,\lambda)}$ are locally unique gapped ground states for $H_{(b,\mu,\lambda)}$ and that they depend continuously on $b$ are the most delicate issues. As mentioned previously, the spectral properties of magnetic Hamiltonians crucially depend on whether the magnetic field strength $b$ is commensurate or not to $2\pi$. Thus, our proof is structured into three parts:
\begin{enumerate}
\item We first treat a  set of commensurate magnetic fields in Section~\ref{sec:commensurate}. For $b$ in the set 
\[ Q:=\{2\pi\,\tfrac{p}{q}\,|\, p\in \Z\,,\; q\in\{12n+1|n\in\N\}\} \subset 2\pi\Q\,, \]
which is dense in $\mathbb{R}$, we consider periodic finite-volume approximations of the system, obtained as restrictions to an increasing sequence of  boxes $\Lambda_{k_n} \subset \Z^2$, $n \in \N$. We show that $\w_{(b,\mu,\lambda)}$ can be obtained as the limit of finite-volume ground states of the form $\w_{(b,\mu,\lambda)}^{(n)}  = \w_{(b,\mu,0)}^{(n)} \circ \alpha_{(b,(0, \lambda))}^{(n)}$, by showing that both $\w_{(b,\mu,0)}^{(n)}$ and $\alpha_{(b,(0, \lambda))}^{(n)}$ converge to the appropriate infinite-volume objects in \eqref{omegadef}. Thus, by Lemma \ref{lem:gaplimit}, $\w_{(b,\mu,\lambda)}$ is a locally unique gapped ground state of $H_{(b,\mu,\lambda)}$.
\item For $b\in \R\setminus Q$, we consider sequences $(b_n)_{n\in\N}$ in $Q$ such that $\lim_{n\to\infty} b_n = b$, and show in Section~\ref{sec:incommensurate} that $\w_{(b,\mu,\lambda)}$ can be obtained as the limit of $\w_{(b_n,\mu,0)} \circ \alpha_{(b_n,(0, \lambda))}$, by again showing that $\w_{(b_n,\mu,0)}$ and $\alpha_{(b_n,(0, \lambda))}$ converge to $\w_{(b,\mu,\lambda)}$ and $\alpha_{(b,(0, \lambda))}$ respectively. Then the same Lemma \ref{lem:gaplimit} allows us to conclude that also here $\w_{(b,\mu,\lambda)}$ is a locally unique gapped ground state of $H_{(b,\mu,\lambda)}$.
\item The third and last part of the proof concludes the weak*-continuity of the map $B_{\delta}(b_0,\mu_0,0) \to S(\mA) , \, (b,\mu,\lambda)\mapsto \w_{(b,\mu,\lambda)}$ using the convergence properties with respect to the magnetic field established in the second part.
\end{enumerate}

\begin{lemma}[{\cite[Lemma~G.1]{WMMMT24}}] \label{lem:gaplimit} 
Let either $\tilde{\mA}_n=\mA_{\Lambda_{k_n}}$ for some strictly increasing sequence of boxes $\Lambda_{k_n}$, i.e.\ $k:\N\to\N$, $n\mapsto k_n$, is strictly increasing, or $\tilde{\mA}_n=\mA$ for all $n\in\N$.
Let $(\mL_{H_n})_{n\in\N}$ be a sequence of derivations, each $\mL_{H_n}$ acting on $\tilde{\mA}_n$, and $(\w_n)_{n\in\N}$ be a sequence of states on $\mA$ such that for some $g>0$ and for all $n\in\N$ and all $A\in \tilde{\mA}_n\cap\mA_0$, 
  \begin{equation}\label{eq:gap2}
  \w_n( A^* \mL_{H_n} A ) \geq g \left( \w_n(A^* A) - \left| \w_n(A) \right|^2 \right)\,.
    \end{equation}
 If  there exist a derivation $\mL_H$ and a state $\w$  on $\mA$ such that for all $A\in\mA_0$ it holds that
 \begin{enumerate}[label={(\roman*)}, ref={(\roman*)}]
 \item \label{item:convergenceLH} $\lim_{n\to \infty}\|(\mL_{H_n}-\mL_H)A\|=0$ and 
 \item \label{item:convergenceGS} $\lim_{n\to \infty}\w_n(A)= \w(A)$,
 \end{enumerate}
 then $\w$ is a locally unique gapped ground state of $\mL_H$ with gap $g$.
\end{lemma}

\subsection{A dense set of commensurate magnetic fields} \label{sec:commensurate}

In this first part of the proof we will show the following statement: For all $(b_0,\mu_0)\in \R^2 $ with $\mu_0\notin \sigma(\mathfrak{h}^{b_0})$ there exist $\delta,g_*,\lambda_0>0$ such that for all $(b,\mu)\in B_{\delta}(b_0,\mu_0)$ with $b\in Q$ and all $\lambda \in (-\lambda_0,\lambda_0)$ the state $\w_{(b,\mu,\lambda)}$, defined as in \eqref{omegadef} with respect to $g_*$, is a locally unique gapped ground state of $H_{(b,\mu,\lambda)}$ with gap $g_*$.

Let $(b_0,\mu_0)\in \R^2$ with $\mu_0\notin \sigma(\mathfrak{h}^{b_0})$.
Since the map  $b\mapsto\sigma(\mathfrak h^b)$ is continuous with respect to the Hausdorff metric, following for instance from gauge-covariant magnetic perturbation theory and the exponential decay of $h$ (see \cite{cornean2010lipschitz}), the spectral gap around $\mu_0$ persists uniformly in a neighborhood of $(b_0,\mu_0)$. Thus there exist $\delta>0$ and $g_*>0$ such that for all $(b,\mu)\in B_\delta(b_0,\mu_0)$ it holds that $\mathrm{dist}(\mu,\sigma(\mathfrak{h}^b))>2 g_*$. 

Now let $(b,\mu)\in B_\delta(b_0,\mu_0)$ with $b\in Q$. This means that $ b=2\pi\frac{p}{q}$ for some  $p\in\Z$ and some   $q\in \{12n+1|n\in\N\}$. 
Thus, magnetic translations with respect to vectors $\gamma$ in the sublattice $\Z^2_q:= \Z\times (q \Z)\subset \Z^2$ are ordinary translations, since the argument $b y_1 \gamma_2= by_1 m q= 2\pi p y_1 m\in 2\pi\Z$ of the phase in the definition of $t^b$ resp.\ $T^b$ (see \eqref{magtrans}) is an integer multiple of $2\pi$ for all $y\in\Z^2$ and $m\in\Z$. This allows us to apply Bloch--Floquet theory to $\mathfrak{h}^b$ with respect to the sublattice $\Z^2_q$ in order to relate the spectrum of $\mathfrak{h}^b$ to the spectra of its restrictions~$\mathfrak{h}_n^b$ to suitable finite boxes~$\Lambda_{k_n}  = \{x\in \Z^2 \, | \, \norm{x}_\infty \leq k_n \}$.  See also Appendix~\ref{app:BF} for details.

Next we define the periodic finite-volume restrictions of the one- and many-body Hamiltonians.
For $n\in\N$, let $k_n:= \frac{1}{2}((12n+1)q-1) $, so that $k_n$ is divisible by 6 and the box's linear size $L_n:= 2k_n+1= (12n+1)q$ is an integer multiple of $q$. The divisibility by 6 is assumed in order to define sub-boxes of size $k_n/2$
and $k_n/3$ during the proof without having to worry about appropriate rounding.
Let $t^b|_{\Lambda_{k_n}}$ resp.\ $T^b|_{\Lambda_{k_n}}$ denote the magnetic translation acting on $\ell^2(\Lambda_{k_n})$ resp.\  on $\mA_{\La_{k_n}}$ with periodic boundary conditions.
More precisely,  for $\g \in \Z^2$ we define 
$t^b|_{{\Lambda_{k_n},\gamma}}$ by 
\[
t^b|_{{\Lambda_{k_n}},\g}: \ell^2(\Lambda_{k_n},\C^\fd) \to \ell^2(\Lambda_{k_n},\C^\fd)\,,\quad  (t^b|_{{\Lambda_{k_n}},\g} f)(y) := \mathrm{e}^{-\mathrm{i}b y_1 \gamma_2} f([y+\gamma]_{\Lambda_{k_n}})
\]
and 
$T^b|_{{\Lambda_{k_n}},\g}$ as the unique automorphism of $\mA_{\La_{k_n}}$ such that 
\begin{align*}
T^b|_{{\Lambda_{k_n}},\g} (a_{y,j}) = \mathrm{e}^{-\mathrm{i}b y_1 \gamma_2} a_{[y+\gamma]_{\Lambda_{k_n}},j} \quad \forall\, y\in \Lambda_{k_n},\, j \in \{1,\dots \fd\},
\end{align*}
where $[y]_{\Lambda_{k_n}}$ is the unique element of $\Lambda_{k_n}$ such that $y-[y]_{\Lambda_{k_n}} \in (L_n\Z)^2$.

The finite volume Hamiltonian $\mathfrak{h}_{ n}^b \colon \ell^2(\Lambda_{k_n},\C^\fd) \to \ell^2(\Lambda_{k_n},\C^\fd)$ has integral kernel 
    \begin{align*}
\mathfrak{h}_{n}^b(x,y) \coloneq 
    \sum_{z\in(L_n\Z)^2 }
    \mathfrak{h}^b(x,y+z) \,.
    \end{align*}
In order to express this   as a sum of $t^b|_{\La_{k_n}}$-translates 
    define 
    \[
    R_{n,0} :=\left\{(y,z)\in \La_{k_n/2}\times\La_{k_n/2}\,|\,\lfloor \tfrac{y_j+z_j+1}{2} \rfloor =0\,,\;j\in\{1,2\}\right\}
    \]and $\mathfrak{h}^b_{n,0}$ to be the operator with kernel
\[
\mathfrak{h}^b_{n,0}(y,z) := \delta_{(y,z)\in R_{n,0}} \mathfrak{h}^b_n(y,z)\,.
\]
Then $\mathfrak{h}^b_n = \sum_{x\in\La_{k_n}}(t^b|_{\La_{k_n},x})^{-1}\mathfrak{h}^b_{n,0}t^b|_{\La_{k_n},x}$ 
and the restriction of  $H_{(b,\mu,\lambda)} $ to $\mA_{\La_{k_n}}$ is defined by
    \begin{align*}
H_{(b,\mu,\lambda),n}  \coloneq \sum_{\gamma\in\Lambda_{k_n}} T^b|_{\Lambda_{k_n}, \gamma}  \left( \mathrm{d}\Gamma(\mathfrak{h}^b_{n,0}) - \mu n_0 +  \lambda\,\E_{\Lambda_{k_n}}(v^b) \right)=:\mathrm{d}\Gamma(\mathfrak{h}^b_n - \mu \mathbf{1}) + \lambda \,V^b_n \,.
    \end{align*}
These restricted Hamiltonians are $t^b|_{\Lambda_{k_n}}$ resp.\ $T^b|_{\Lambda_{k_n}}$-invariant and by construction it holds for all $n\in\N$ that $\sigma(\mathfrak{h}^b_n) \subset \sigma(\mathfrak{h}^b) $, see Appendix~\ref{app:BF} for details. 

We will use the fact that the dynamics generated by these restricted Hamiltonians approximates the one in infinite-volume in an appropriate sense, which is a consequence of the following general lemma.

\begin{restatable}{lemma}{lemfinitevolapprox}
\label{lem:finite_vol_approx}
    Let $K\subseteq 3\,\N_0$ be some infinite index set, $T$ a translation and for $k\in K$ let $T|_{\Lambda_k}$ be its restriction to the box $\Lambda_k$, defined via
    \begin{align*}
        T|_{\Lambda_k,\gamma}(a_{y,j})  = T_{([y+\gamma]_{\Lambda_k} -y)}(a_{y,j}) \quad \forall\, y\in \Lambda_{k},\, j \in \{1,\dots \fd\}. 
    \end{align*}
    Further, let $I\subseteq \R$ be an interval, $(\Phi^t)_{t\in I}$ a continuous family of $T$-invariant $B_\infty$-interactions and for each $k\in K$ let $(\Phi^{t,k}_0)_{t\in I}$ be a continuous family of self-adjoint $T|_{\Lambda_k}$-compatible elements of $\mA^N_{\Lambda_k} := \mA_{\Lambda_k} \cap \mA^N$. Define the corresponding family of $T|_{\Lambda_k}$-invariant operators as 
    \[
        \Phi^{t,k} = \sum_{\gamma\in\Lambda_k} T|_{\Lambda_k,\gamma}\Phi^{t,k}_0\,.
    \]
 Denote the associated dynamics on $\mA$ resp.\ $\mA_{\Lambda_k}$ by 
$(\alpha_{s,t})_{s,t\in I}$ and $(\alpha^{k}_{s,t})_{s,t\in I}$.

    Assume that for all $\nu\in\N_0$ it holds that 
    \[
  \lim_{k\to\infty}
    \sup_{u\in I} \norm{ 
    \Phi^u_0 - \Phi^{u,k}_0}_{\nu}=0\,.
    \]
   Then  for all $\nu\in\N_0$
   \begin{align*}
        \sup_{t\in I}\sup_{A\in \mA_\infty\setminus\{0\}} \frac{\norm{\mL_{\Phi^t} A - \mL_{\Phi^{t,k}} \E_{\La_{k}}A }_{\nu}}{\norm{A}_{\nu+4}} \xrightarrow{k\to\infty} 0
    \end{align*}
    and there exists an increasing function $f_\nu:[0,\infty)\to (0,\infty)$, growing at most polynomially  at infinity, such that 
    \begin{align*}
       \sup_{t,s\in I} \sup_{A\in \mA_\infty\setminus\{0\}} \frac{\norm{\alpha_{s,t} A - \alpha_{s,t}^{k} \E_{\La_{k}}A }_{\nu}}{f_\nu(|t-s|) \,\norm{A}_{\nu+4}} \xrightarrow{k\to\infty} 0\,.
    \end{align*}
\end{restatable}
\noindent
See Appendix~\ref{proof:finite_vol_approx} for the proof of Lemma~\ref{lem:finite_vol_approx}.

\vspace{0.5\baselineskip}

To see that $H_{(b,\mu,\lambda)}$ and its finite volume restrictions 
$H_{(b,\mu,\lambda),n}$ satisfy the assumptions of
Lemma~\ref{lem:finite_vol_approx}, note that for all $\nu \in \N_0$
\[
\norm{ v^b - \E_{\La_{k_n}}v^b}_{\nu} \xrightarrow{n\to \infty} 0
\]
and 
\[ 
\norm{  \dd\Gamma(\mathfrak{h}^b_0) -  \dd\Gamma(\mathfrak{h}^b_{n,0}) }_{\nu}
 \leq \norm{\dd\Gamma(\mathfrak{h}^b_0)  - \E_{\La_{k_n/2}} \dd\Gamma(\mathfrak{h}^b_0)  }_{\nu}+ 
\norm{\E_{\La_{k_n/2}} \dd\Gamma(\mathfrak{h}^b_0) -  \dd\Gamma(\mathfrak{h}^b_{n,0}) }_{\nu}\,.
 \]
 Here the first term vanishes for $n\to\infty$ since $\dd\Gamma(\mathfrak{h}^b_0)\in\mA_\infty$, and for the second term we find
\begin{align*}
 \hspace{2em}&\hspace{-2em}   
    \norm{\E_{\La_{k_n/2}} \dd\Gamma(\mathfrak{h}^b_0) -  \dd\Gamma(\mathfrak{h}^b_{n,0})  }_{\nu}
 \\
 &\leq
 \sum_{x,y\in\La_{k_n/2}} \delta_{(x,y)\in R_{n,0}}  \sum_{i,j =1}^\fd \norm*{ a^*_{x,i} \, (\mathfrak{h}^b(x,y)_{i,j}-\mathfrak{h}^b_n(x,y)_{i,j})\,  a_{y,j}}_{\nu} 
 \\
 &\leq \sum_{x,y\in\La_{k_n/2}} \delta_{(x,y)\in R_{n,0}} 
 \sum_{i,j =1}^\fd \sum_{z\in(L_n\Z)^2\setminus\{0\}}\norm*{ a^*_{x,i} \, (\mathfrak{h}^b(x,y+z)_{i,j})\,  a_{y,j}}_{\nu}
 \\
 &\leq
 \sum_{x,y\in\La_{k_n/2}} \delta_{(x,y)\in R_{n,0}} 
 \sum_{i,j =1}^\fd \sum_{z\in(L_n\Z)^2\setminus\{0\}}
  3\, \left(\tfrac{L_n}{2}\right)^\nu\, \abs{\mathfrak{h}^b(x,y+z)_{i,j}}
  \\
 &\leq 
 3\, \left(\tfrac{L_n}{2}\right)^\nu \,\fd^2\,\sum_{x,y\in\La_{k_n/2}} \delta_{(x,y)\in R_{n,0}} 
\sum_{z\in(L_n\Z)^2\setminus\{0\}}
    C\,   \e^{-c\|x-y-z \|}  
\\
&\leq
3\, \left(\tfrac{L_n}{2}\right)^\nu \,\fd^2\,\sum_{x,y\in\La_{k_n/2}} \delta_{(x,y)\in R_{n,0}} 
\sum_{z\in(L_n\Z)^2\setminus\{0\}}
    C\,   \e^{c\|x-y\|} \e^{-c\|z \|}  
\\
&\leq
3\, \left(\tfrac{L_n}{2}\right)^\nu \,\fd^2\,C\, L_n^2\,\e^{cL_n/2}\,
\sum_{z\in(L_n\Z)^2\setminus\{0\}}
     \e^{-c\| z\|}  
\\
&\leq
3\, \left(\tfrac{L_n}{2}\right)^\nu \,\fd^2\,C\, L_n^2\,\e^{cL_n/2}\,
  \frac{4\e^{-cL_n}}{(1- \e^{\tfrac{-cL_n}{2}})^2} \; \xrightarrow{n\to\infty}\;0  \,.     
\end{align*}
Hence, Lemma~\ref{lem:finite_vol_approx} yields that
\begin{equation} \label{eq:dynamics_converges}
        \sup_{s \in \R} \sup_{A\in \mA_\infty\setminus \{0\}} \frac{\lVert \e^{-\i s \mL_{H_{(b,\mu,\lambda),n}}}\E_{\La_{k_n}} A - \e^{-\i s \mL_{H_{(b,\mu,\lambda)}}} A \rVert_\nu}{f_\nu(|s|)\,\lVert A \rVert_{\nu+4}}\,  \xrightarrow{n\to \infty} 0\,.
\end{equation}

As discussed in detail  in \cite{WMMMT24}, the result of \cite{de2019persistence} implies  that there exists $\lambda_0>0$ such that for $|\lambda|<\lambda_0$  and all $n\in\N$  the operator $H_{(b,\mu,\lambda),n} $ has a unique $T^b|_{\Lambda_{k_n}}$-invariant ground state $\w_{(b,\mu,\lambda)}^{(n)}$ with  gap at least equal to $\frac{1}{2}\,\mathrm{dist}(\mu, \sigma(\mathfrak{h}_{n}^{b})  )$, which, by the spectral inclusion $\sigma(\mathfrak{h}^b_n) \subset \sigma(\mathfrak{h}^b) $,  is larger or equal to $\frac{1}{2}\, \mathrm{dist}(\mu, \sigma(\mathfrak{h}^{b})  )>g_*$ uniformly in $n$.  Moreover, this $\lambda_0$ can be chosen uniformly for all $(b,\mu)\in B_\delta(b_0,\mu_0)\cap(Q\times\R)$, as it depends only on interaction norms of $H_{(b,\mu,0)}$ and   $v^b$  that are uniformly bounded for all $(b,\mu)\in B_\delta(b_0,\mu_0)\cap(Q\times\R)$, see \cite{de2019persistence}.

For $\lambda=0$ the unique ground states of $H_{(b,\mu,0),n}$ resp.\ $H_{(b,\mu,0)}$ are the quasi-free states with two-point functions given by the summation kernel of $P_{(-\infty,\mu)}^{\mathfrak{h}^b_{n}}$, the Fermi projection of $\mathfrak{h}_{n}^b$, resp.\ of $P_{(-\infty,\mu)}^{\mathfrak{h}^b}$, the Fermi projection of $\mathfrak{h}^b$.
Automorphic equivalence in finite volume \cite{BMNS12} now implies that
\[
\w_{(b,\mu,\lambda)}^{(n)} = \w_{(b,\mu,0)}^{(n)} \circ \alpha_{(b,(0, \lambda))}^{(n)},
\]
where $\alpha_{(b,(0, \lambda))}^{(n)}$ is the unique solution of 
\[
\partial_\lambda \,\alpha_{(b,( 0, \lambda))}^{(n)} = \alpha_{(b,( 0, \lambda))}^{(n)} \circ \i \mL_{K^{(b,\lambda),n}}\,,
\quad \alpha_{(b,(0, 0))}^{(n)} = \mathbf{1}\,,
\]
with 
\[
K^{(b,\lambda),g_*,n}_0   := 
-  \int_\R \D s\, W_{g_*}(s)
\,\e^{\i s\mL_{H_{(b,\mu,\lambda),n}}} \E_{\La_{k_n}}(v^b)
\]
and
\[ 
K^{(b,\lambda),g_*,n}   := 
 \sum_{\gamma\in\Lambda_{k_n}}  T^b|_{\Lambda_{k_n}, \gamma}  K^{(b,\lambda),g_*,n}_0 
\,.
\]
  To apply Lemma~\ref{lem:finite_vol_approx} to   $K^{(b,\lambda),g_*,n}$, we recall the definition of $K^{(b,\lambda),g}$ in \eqref{KDef} and observe that with \eqref{eq:dynamics_converges} we have
\begin{align*}
    \hspace{2em}&\hspace{-2em}
    \lVert    K^{(b,\lambda),g_*,n}_0 -  K^{(b,\lambda),g_*}_0  \rVert_\nu
    \\
    &\leq
    \int_{\R}\mathrm{d}s \, \abs{W_{g_*}(s)}  \left\|
    \e^{\i s\mL_{H_{(b,\mu,\lambda),n}}}  \, \E_{\La_{k_n}} v^b - \e^{\i s\mL_{H_{(b,\mu,\lambda)}}}   \, v^b 
    \right\|_\nu  
    \\
    &\leq
    \sup_{u\in\R}\frac{\lVert \e^{-\i u \mL_{H_{(b,\mu,\lambda),n}}}\E_{\La_{k_n}} v^b - \e^{-\i u \mL_{H_{(b,\mu,\lambda)}}} v^b \rVert_\nu}{f_\nu(|u|)\,\lVert v^b \rVert_{\nu+4}}
    \int_{\R}\mathrm{d}s \, \abs{W_{g_*}(s)} \,f_\nu(|s|)\, \|v^b\|_{\nu+4} 
    \\
    &\stackrel{n\to\infty}{\rightarrow} 0\,.
\end{align*}
 Thus
\begin{align}\label{eq:Auto_convergence}
        \sup_{\lambda \in [-\lambda_0,\lambda_0]} \sup_{A\in \mA_\infty\setminus \{0\}} \frac{\lVert \alpha_{(b,(0, \lambda))}^{(n)} \E_{\La_{k_n}} A - \alpha_{(b,(0, \lambda))} A \rVert_\nu}{\lVert A \rVert_{\nu+4}}\,  \xrightarrow{n\to \infty} 0\, .
\end{align}

To show that also the sequence of states $\w_{(b,\mu,\lambda)}^{(n)}$ converges, it remains to check this  for $\lambda=0$:
Using Wick's rule it follows that for any fixed finite region $\Lambda\subset\Z^2$ 
\begin{align*}
    \hspace{2em}&\hspace{-2em}
    \lim_{n\to\infty} \sup_{x,y\in\Lambda} \left\|P_{(-\infty,\mu)}^{\mathfrak{h}^b_{n}}(x,y) - P_{(-\infty,\mu)}^{\mathfrak{h}^b}(x,y)\right\| =0 
    \\
    &\quad\Rightarrow\quad 
    \forall A\in\mA_\Lambda\quad \lim_{n\to\infty}
    \left|\left(\w_{(b,\mu,0)}^{(n)}-\w_{(b,\mu,0)}\right)(A)\right|=0\,.    
\end{align*}
The premise in the above implication follows  from the aliasing formula for Fourier series: since we have $(L_n\Z)^2 $-periodicity,
\begin{align*}
    P_{(-\infty,\mu)}^{\mathfrak{h}^b_{n}}(x,y) 
    &= 
    \sum_{m\in \Z^2} P_{(-\infty,\mu)}^{\mathfrak{h}^b}(x,y-mL_n ) \numberthis \label{eq:fermi_restriction} 
    \\ 
    &=
    P_{(-\infty,\mu)}^{\mathfrak{h}^b}(x,y) + \sum_{m\in\Z^2\setminus\{0\}} P_{(-\infty,\mu)}^{\mathfrak{h}^b}(x,y-mL_n)\,,    
\end{align*}
see Appendix~\ref{app:BF} for details.
Bloch--Floquet theory, reviewed in the same Appendix, also shows exponential decay of the kernel away from the diagonal, i.e.\ there are constants $c,C>0$ such that 
\begin{equation} \label{eqn:decay_of_kernel}
\forall x,y\in\Z^2\,:\qquad \left\|P_{(-\infty,\mu)}^{\mathfrak{h}^b}(x,y )\right\|\leq C\,\e^{-c\|x-y\|}\,,
\end{equation}
and thus
\[ \lim_{n\to\infty} \sup_{x,y\in\Lambda} \left\|P_{(-\infty,\mu)}^{\mathfrak{h}^b_{n}}(x,y) - P_{(-\infty,\mu)}^{\mathfrak{h}^b}(x,y)\right\| =0 \]
follows.
Hence, for all $A\in\mA_0$ we have $\lim_{n\to\infty}
|(\w_{(b,\mu,0)}^{(n)}-\w_{(b,\mu,0)})(A)|=0$, and by density of $\mA_0$ in $\mA$ and boundedness of states, actually  for all $A\in\mA$   we have  
\begin{equation}\label{eq:statecobv}
    \lim_{n\to\infty}
|(\w_{(b,\mu,0)}^{(n)}\circ \E_{\Lambda_{k_n}}-\w_{(b,\mu,0)})(A)|=0\,.
\end{equation}

Finally, putting everything together, we obtain for $A\in\mA_0$ and $n\in \N$ such that $A \in \mA_{\Lambda_{k_n}}$
\begin{eqnarray*}
    \left|\w_{(b,\mu,\lambda)}^{(n)}(A)-\w_{(b,\mu,\lambda)}(A)\right|
    &\leq &
    \left|\w_{(b,\mu,0)}^{(n)}\circ \E_{\Lambda_{k_n}}\,  \left( \alpha_{(b,(0,\lambda))}^{(n)}\,A - \alpha_{(b,(0,\lambda))} A\right) \right|  
    \\
    &&+\;
    \left|\left( \w_{(b,\mu,0)}^{(n)}\circ \E_{\Lambda_{k_n}} - \w_{(b,\mu,0)}\right)  \alpha_{(b,(0,\lambda))} (A)  \right|
    \\
    &\leq &
    \left\| \alpha_{(b,(0,\lambda))}^{(n)}(A) -\alpha_{(b,(0,\lambda))}(A)\right\| 
    \\
    &&+\;
    \left|\left(\w_{(b,\mu,0)}^{(n)} \circ \E_{\Lambda_{k_n}} - \w_{(b,\mu,0)} \right) \circ \alpha_{(b,(0,\lambda))} (A)\right|
    \\
    &\stackrel{n\to\infty}{\longrightarrow}&0\,,
\end{eqnarray*}
where we used \eqref{eq:Auto_convergence} for the first term and \eqref{eq:statecobv}
for the second.

 For local  $A\in\mA_0$ we also have $ \|(\mathcal L_{H_{(b,\mu,\lambda),n}}-\mathcal L_{H_{(b,\mu,\lambda)}})A\|\to0$ because of Lemma~\ref{lem:finite_vol_approx}, and the finite-volume ground states satisfy the uniform gap inequality. Hence Lemma~\ref{lem:gaplimit} implies that $\omega_{(b,\mu,\lambda)}$ is a locally unique gapped ground state with gap at least $g_*$.

\subsection{All other  magnetic fields} \label{sec:incommensurate}

This part of the proof extends the statement proven in the first part by removing the restriction that $b$ must lie in $Q$. For this let $(b_0,\mu_0)\in \R^2 $ with $\mu_0\notin \sigma(\mathfrak{h}^{b_0})$. By the statement of Section \ref{sec:commensurate} there exist $\delta,g_*,\lambda_0>0$ such that for all $(b,\mu)\in B_{\delta}(b_0,\mu_0)$ with $b\in Q$ and all $\lambda \in (-\lambda_0,\lambda_0)$ the state $\w_{(b,\mu,\lambda)}$, defined as in \eqref{omegadef} with respect to $g_*$, is a locally unique gapped ground state of $H_{(b,\mu,\lambda)}$ with gap $g_*$.

Now let $(b,\mu)\in B_{\delta}(b_0,\mu_0)$  and $\lambda \in (-\lambda_0,\lambda_0)$. We pick a sequence $(b_n)_{n\in\N}$ in $Q$ with $\lim_{n\to\infty}b_n=b$ and $(b_n,\mu)\in B_{\delta}(b_0,\mu_0)$, for all $n\in \N$. We know that for each $n\in \N$ the state $\w_{(b_n,\mu,\lambda)}$ is a locally unique gapped ground state of $H_{(b_n,\mu,\lambda)}$ with gap $g_*$. To conclude  that also $\w_{(b,\mu,\lambda)}$ is a locally unique gapped ground state of $H_{(b,\mu,\lambda)}$ with gap $g_*$, we employ Lemma~\ref{lem:gaplimit}. For this we need to check that for all $A\in\mA_0$ we have
\begin{enumerate}[label={(\roman*)}, ref={(\roman*)}]
\item \label{eq:b_conv_H} $\lim_{n\to\infty}\|\mL_{H_{(b_n,\mu,\lambda)}} A - \mL_{H_{(b,\mu,\lambda)}} A\|=0$ and 
\item \label{eq:b_conv_omega} $\lim_{n\to\infty} \w_{(b_n,\mu,\lambda)}(A) =  \w_{(b,\mu,\lambda)}(A).$
\end{enumerate}
We prove 
  the above statements for arbitrary sequences $(b_n)_{n\in\N}$ in $\R$ with $b_n\to b$ and thus, in particular, also establish continuity of the maps $b\mapsto \w_{(b,\mu,\lambda)}(A)$. We'll resort to the following two intermediate results, whose proofs are postponed to Appendix~\ref{sec:convergences}.

\begin{restatable}{lemma}{lemconvergenceofcocycles}
\label{lem: convergence of cocycles}
    Let $I \subseteq \R$ be an interval and for $n\in \N\cup \{\infty\}$ let $T^n$ be a translation and $(\alpha_{(s,t),n})_{s,t\in I}$  the cocycle generated by the continuous $T^n$-invariant family of $ B_\infty$-interactions $(\Phi^{t,n})_{t\in I}$. Assume that for all   $ \nu \in \N_0 $ and $ \g \in \Z^2 $ it holds that
    \begin{align*}
        \sup_{s\in I}\lVert \Phi_0^{s,n} - (T_\g^n)^{-1}T_\g^\infty \,  \Phi^{s,\infty}_0 \rVert_{\nu} \xrightarrow{n \to \infty} 0\, .
    \end{align*}
    Then for   all $\nu \in \N_0$
    \begin{align*}
       \sup_{t\in I} \sup_{A\in \mA_\infty\setminus\{0\}} \frac{\norm{\mL_{\Phi^{t,n}} A - \mL_{\Phi^{t,\infty}} A }_{\nu}}{\norm{A}_{\nu+3}} \xrightarrow{n\to\infty} 0\,.
    \end{align*}
    and there exists an increasing function $f_\nu:[0,\infty)\to (0,\infty)$, growing at most polynomially at infinity, such that 
    \begin{align*}
        \sup_{s,t \in I} \sup_{A\in \mA_\infty\setminus \{0\}} \frac{\lVert \alpha_{(s,t),n} A - \alpha_{(s,t),\infty} A \rVert_\nu}{ f_\nu(|t-s|) \,\lVert A \rVert_{\nu+3}}\,  \xrightarrow{n\to \infty} 0\, .
    \end{align*}
\end{restatable}
 \noindent  See Appendix~\ref{proof: convergence of cocycles} for the proof of Lemma~\ref{lem: convergence of cocycles}.

\begin{restatable}{lemma}{lemtranslatingbackandforth}\label{lem: translating back and forth}
    Let $b_0, b \in \R$ and $\gamma \in \Z^2$. It holds that 
    $$ (T^{b}_\g)^{-1} \, T^{b_0}_\g = \e^{\i \g_2(b_0-b)\mL_{X_1}}$$
    and for each $\nu \in \N_0$ there is a $c_\nu>0$ such that for all $A\in \mA_\infty$
    \begin{align*}
        \lVert(T^{b}_\g)^{-1} \, T^{b_0}_\g \, A - A \rVert_\nu \leq c_\nu \, \lvert b_0 - b \rvert \, \lvert \g_2 \rvert \, \lVert A \rVert_{\nu+4}\, .
    \end{align*}
\end{restatable}
\noindent See Appendix~\ref{proof:translating back and forth} for the proof of Lemma~\ref{lem: translating back and forth}.

\medskip

With these results at hand, we see that~\ref{eq:b_conv_H} holds due to Lemma \ref{lem: convergence of cocycles} and the fact that, by Lemma~\ref{lem: translating back and forth},
\begin{eqnarray}\label{eq:H_compare}
\|H_{(b_n,\mu,\lambda),0} - (T^{b_n}_\gamma )^{-1}T^{b}_\gamma H_{(b,\mu,\lambda),0}\|_\nu &\leq&
\|H_{(b_n,\mu,\lambda),0} -H_{(b,\mu,\lambda),0}\|_\nu
\nonumber\\&&+\; \|H_{(b,\mu,\lambda),0} - (T^{b_n}_\gamma )^{-1}T^{b}_\gamma H_{(b,\mu,\lambda),0}\|_\nu
\nonumber\\
&\leq&
\|H_{(b_n,\mu,\lambda),0} -H_{(b,\mu,\lambda),0}\|_\nu
\nonumber\\&&+\;c_\nu \,  |b_n-b| \,|\gamma_2|\,\|H_{(b,\mu,\lambda),0}\|_{\nu+4}\,,
\end{eqnarray}
i.e., that the local terms of $H$ converge uniformly on bounded sets.

Coming to \ref{eq:b_conv_omega}, we again compare the states at $\lambda=0$ and the automorphisms separately:
\begin{eqnarray}
        |\w_{(b_n,\mu,\lambda)}(A)-\w_{(b,\mu,\lambda)}(A)|
        &\leq& 
        |\w_{(b_n,\mu,0)}\left( \alpha_{(b_n,(0,\lambda))} \,A- \alpha_{(b,(0,\lambda))}\,A\right)| \nonumber\\
        &&+\;|\left( \w_{(b_n,\mu,0)}-\w_{(b,\mu,0)}\right)(\alpha_{(b,(0,\lambda))}\, A) |
        \nonumber\\
        &\leq& \| \alpha_{(b_n,(0,\lambda))}\,A -\alpha_{(b,(0,\lambda))}\,A\|\label{est4}\\
        &&+\;|\left( \w_{(b_n,\mu,0)}-\w_{(b,\mu,0)}\right)(\alpha_{(b,(0,\lambda))}\,A) |\,. \label{est3}
    \end{eqnarray}
For \eqref{est3} we use again that on finite sets $\Lambda$ 
\begin{equation}\label{magneticfermi}
\lim_{n\to\infty} \sup_{x,y\in\Lambda} \left|P_{(-\infty,\mu)}^{\mathfrak{h}^{b_n}}(x,y) - P_{(-\infty,\mu)}^{\mathfrak{h}^b}(x,y)\right| =0\,,
\end{equation}
which implies weak$^*$-convergence of $\w_{(b_n,\mu,0)}$ to $\w_{(b,\mu,0)}$ by the same arguments as in the previous subsection.
Note that \eqref{magneticfermi} can be proved by the Riesz formula for the Fermi projections and resolvent estimates provided by the methods of gauge covariant magnetic perturbation theory, see~\cite{nenciu2006smoothness, cornean2010lipschitz, cornean2019magnetic}. 

It remains to show convergence of \eqref{est4}, for which we use again Lemma~\ref{lem: convergence of cocycles}. 
We find using Lemma~\ref{lem: translating back and forth} that
\begin{align*}
    \hspace{2em}&\hspace{-2em}
    \lVert    K^{(b_n,\lambda),g_*}_0 - (T_\g^{b_n})^{-1}T^{b}_\gamma K^{(b,\lambda),g_*}_0  \rVert_\nu 
    \\
    &\leq
    \lVert    K^{(b_n,\lambda),g_*}_0 - K^{(b,\lambda),g_*}_0 \rVert_\nu + \lVert    K^{(b,\lambda),g_*}_0 - (T_\g^{b_n})^{-1}T^{b}_\gamma K^{(b,\lambda),g_*}_0  \rVert_\nu
    \\
    &\leq
    \lVert    K^{(b_n,\lambda),g_*}_0 - K^{(b,\lambda),g_*}_0 \rVert_\nu + c_\nu \, |b_n-b|\,|\gamma_2|\, \| K^{(b,\lambda) ,g_*}_0\|_{\nu+4}\,    
\end{align*}
and
\begin{align*}
    \hspace{2em}&\hspace{-2em}
    \lVert    K^{(b_n,\lambda),g_*}_0 - K^{(b,\lambda),g_*}_0 \rVert_\nu
    \\
    &\leq
    \int_{\R}\mathrm{d}s \,\abs{W_{g_*}(s)}  \left\| \e^{\i s\mL_{H_{(b_n,\mu,\lambda)}}} \, (v^{b_n} - v^b) \right\|_\nu
    \\
    &\hspace{2em}+
    \int_{\R}\mathrm{d}s \,\abs{W_{g_*}(s)}  \left\| \left(\e^{\i s\mL_{H_{(b_n,\mu,\lambda)}}} -\e^{\i s\mL_{H_{(b,\mu,\lambda)}}}\right) \, v^{b} \right\|_\nu
    \\
    &\leq
    \int_{\R}\mathrm{d}s \, \abs{W_{g_*}(s)} \,f_\nu(|s|)\, \|v^{b_n}- v^b\|_{\nu}
    \\
    &\hspace{2em}+
    \sup_{u\in\R}\frac{\left\lVert  \left( \e^{\i u\mL_{H_{(b_n,\mu,\lambda)}}} -\e^{\i u\mL_{H_{(b,\mu,\lambda)}}} \right) \,v^{b} \right\rVert_\nu}{f_\nu(|u|)\,\lVert v^{b} \rVert_{\nu+4}}
    \int_{\R}\mathrm{d}s \, \abs{W_{g_*}(s)} \,f_\nu(|s|)\, \|v^{b}\|_{\nu+4}    
    \\
    &\stackrel{n\to\infty}{\rightarrow} 0\,.
\end{align*}

\subsection{Continuity of $(b,\mu,\lambda)\mapsto \w_{(b,\mu,\lambda)}$} \label{continuity_in_lambda}
Let $b_0$, $\mu_0$, $\delta$ and $\lambda_0$ as in Section \ref{sec:incommensurate}. Continuity of the map \[B_\delta(b_0,\mu_0)\times (-\lambda_0,\lambda_0),\, (b,\mu,\lambda)\mapsto \w_{(b,\mu,\lambda)}(A)\]
for all $A\in\mA$ now follows easily: first note that since all states have norm one, it suffices to prove continuity of the map for $A$ in a dense subset of $\mA$, e.g.\ for $A\in\mA_0$.
Second, we already observed that $\w_{(b,\mu,\lambda)} = \w_{(b,\mu_0,\lambda)}$
for all $(b,\mu,\lambda)\in B_\delta(b_0,\mu_0)\times (-\lambda_0,\lambda_0)$. Finally, continuity of 
$b\mapsto \w_{(b,\mu, \lambda)}(A)$ for $A\in\mA_0$ was proved in the previous subsection.

For all $(b,\mu) \in \R^2$ and all $\tilde{\lambda}, \lambda \in \R$ 
\begin{eqnarray*}
|\w_{(b,\mu,\tilde \lambda)}(A)-  \w_{(b,\mu,\lambda)}(A)| &=&
\left|\w_{(b,\mu,0)}\left(\alpha_{(b,(0,\tilde \lambda))}\,A-\alpha_{(b,(0, \lambda))}\,A \right)\right|\\
&\leq& \left\|\alpha_{(b,(0,\tilde \lambda))}\,A-\alpha_{(b,(0,\lambda))}\,A\right\|\\
&\leq & \int_{\tilde\lambda\land \lambda}^{\tilde\lambda\lor \lambda} \D u\, \|\mL_{K^{(b,u),g_*}} A\|
\\
&\leq&
|\lambda-\tilde\lambda|\, \sup_{u\in [\tilde\lambda\land \lambda, \tilde\lambda\lor \lambda]} c\,\|K^{(b,u),g_*}\|_{3}\,\|A\|_{5}
\end{eqnarray*}
where $c$ is some constant as in Lemma \ref{lem: sum representation of generator}.
Since the map $(b,u)\mapsto \|K^{(b,u),g_*}\|_{3}$ is bounded uniformly on compact subsets of $\R\times (-\lambda_0,\lambda_0)$, this implies that $\lambda\mapsto \w_{(b,\mu_0,\lambda)}(A)= \w_{(b,\mu_0,0)}\circ \alpha_{(b,(0,\lambda))}(A)$ is uniformly Lipschitz for $b$ in compacts and we can conclude that  $B_\delta(b_0,\mu_0)\times (-\lambda_0,\lambda_0),\, (b,\mu,\lambda)\mapsto \w_{(b,\mu,\lambda)}(A)$ is continuous.

\section{Proof of Theorem \ref{cor:universality}} \label{sec:universality}
 
By definition the family of states $\big(\w_{(b,\mu,\lambda)}\big)_{(b,\mu,\lambda)\in\mathcal{P}_{M_0}}$ satisfies for $(b,\mu,\lambda),(b,\mu,\tilde{\lambda}) \in \mP_{M_0}$ that 
\[
    \w_{(b,\mu,\lambda)} = \w_{(b,\mu,\tilde{\lambda})} \circ \alpha_{(b,(\tilde{\lambda},\lambda))}.
\]
This automorphic equivalence and the Chern--Simons formula (cf.~\cite[Lemma~5.2 and Corollary~5.3]{WMMMT24}) imply that the Hall conductivity \eqref{sigmaHdef} for the state $\w_{(b,\mu,\lambda)}$ is constant in $\lambda$ and since for each $(b,\mu,\lambda)\in \mP_{M_0}$ also the line segment $[(b,\mu,0),(b,\mu,\lambda)]$ is contained in $\mP_{M_0}$, one obtains: for all $(b,\mu,\lambda)\in \mathcal{P}_{M_0}$
\[ \sigma^{\mathrm{H}}_{(b,\mu,\lambda)} = \sigma^{\mathrm{H}}_{(b,\mu,0)} \,.
\]
It remains to show that the non-interacting Hall conductivity 
\[
\sigma^{\mathrm{H}}_{(b,\mu,0)} =
\overline{\w}_{(b,\mu,0)} \left( \i \, \left[ X_2\OD[(b,\mu,0)], X_1\OD[(b,\mu,0)] \right] \right) 
\]
is constant  and  quantized on $M_0$.
This will follow immediately once we show that the many-body formula for $\sigma^\mathrm{H}$ agrees, for non-interacting systems, with the one-body formula in terms of Fermi projections.
This is the content of the following proposition, which we formulate in greater generality also for non-periodic systems. In that case  we use the  Hall conductance defined in terms of so-called switch-functions instead of position operators. Compare  \cite{TW2025} and references therein.

\begin{proposition}\label{prop:dcformulas}
 Let $h\in\mL(\ell^2(\Z^2,\C^\fd))$ be self-adjoint  and $\mu\in\R\setminus\sigma(h)$ such that $g:= \mathrm{dist}(\mu,\sigma(h))>0$.  Let $P:= \chi_{(-\infty,\mu)}(h)$. Assume that the kernel $h(x,y)$ of $h$ has uniform exponential decay, i.e.,
 \begin{equation}\label{hexp}
 \exists c,C>0\;\forall x,y\in\Z^2\quad \|h(x,y)\|\leq C\e^{-c\|x-y\|}\,.
 \end{equation}

\begin{enumerate}
    \item[(i)] 
 Then $H:=\dd\Gamma(h-\mu)\in B_{\exp}$ and the unique quasi-free state $\w$ with $\w(a^*_{x,j}a_{y,i}) = P(x,y)_{ji}$   for all $x,y\in\Z^2$ and $j,i\in\{1,\ldots,\fd\}$ is the unique ground state of $\mL_{H}$ and has a gap of size at least $g$.
\item[(ii)] 
Let $\lambda_j(x):= \chi_{x_j\geq 0}(x)$, $j\in\{1,2\}$,  denote the one-body switch-functions, which define  corresponding multiplication operators $\lambda_j$ on $\ell^2(\Z^2,\C^\fd)$, and $\Lambda_j := \dd\Gamma(\lambda_j)$ their many-body versions. Then
 \[
\i\,\w([\Lambda_2\OD,\Lambda_1\OD]) = \i\,\tr (P [\lambda_2\ODT,\lambda_1\ODT]) \in \tfrac{1}{2\pi} \Z\,,
 \]
 where $\lambda_j\ODT := P\lambda_jP^\perp +P^\perp\lambda_jP $.
\item[(iii)] 
 If $H$ is periodic with respect to some translation $T$, then also
 \[ \overline\w([X_2\OD,X_1\OD])=\w([\Lambda_2\OD,\Lambda_1\OD])\,.
 \]
 \end{enumerate}
\end{proposition}

Before proving the proposition, we conclude the proof of Theorem~\ref{cor:universality}: Proposition~\ref{prop:dcformulas} shows that for all $(b,\mu)\in M_0$  
\[
\sigma^\mathrm H_{(b,\mu,0)}
=
\mathrm i\,\operatorname{tr}
\bigl(
P^{\mathfrak{h}^b}_{(-\infty,\mu)}
[\lambda_2\ODT[(b,\mu)],
\lambda_1\ODT[(b,\mu)]]
\bigr)
\]
with $\lambda_j\ODT[(b,\mu)]:= P^{\mathfrak{h}^b}_{(-\infty,\mu)} \lambda_j (P^{\mathfrak{h}^b}_{(-\infty,\mu)})^\perp +  (P^{\mathfrak{h}^b}_{(-\infty,\mu)})^\perp \lambda_j P^{\mathfrak{h}^b}_{(-\infty,\mu)}$.
That this expression for $\sigma^\mathrm H_{(b,\mu,0)}$ is continuous in $(b,\mu)\in M_0$ can be proved with the methods of magnetic perturbation theory, in view of the uniform gap: compare \cite[Sec.~2.1.2]{cornean2021beyond}. Since it takes values in the discrete set $\frac1{2\pi}\mathbb Z$, it is locally constant. Because $M_0$ is connected, $(b,\mu)\mapsto \sigma^\mathrm H_{(b,\mu,0)}$ is constant on $M_0$.

\begin{proof}[Proof of  Proposition~\ref{prop:dcformulas}]
Most of the statements are well known or follow from well known facts about second quantization. The simple argument for (i) is given for example  in \cite{WMMMT24}. (iii) is established in \cite{TW2025}. In (ii) the inclusion $\tr (P [\lambda_2\ODT,\lambda_1\ODT]) \in \tfrac{\i}{2\pi} \Z$ is also well known, see e.g.\ \cite{avron1985quantization, rammal1990algebraic, bellissard1994noncommutative, aizenman1998localization, cornean2021beyond, cornean2025magnetic}.    

The only statement in need of a proof is the equality of the many-body and the one-body double-commutator formulas, i.e.\ that 
$\w([\Lambda_2\OD,\Lambda_1\OD]) = \tr (P [\lambda_2\ODT,\lambda_1\ODT])$.
To this end note that, for self-adjoint one-body operators $a\in\mL(\ell^2(\Z^2,\C^\fd))$, 
it holds that  $A:=\dd\Gamma(a)$ is in $\mA$ iff $a$ is trace class (see, e.g., \cite{Lundberg}), and in  that case  we have
\[
\w(A) = \tr(P a)\,.
\]
Moreover,
\begin{eqnarray*}
    (\Lambda_j\OD)_x 
    &\coloneq &
    \i \,  \int_{\R}\dd s \, W_g(s)\, \e^{\i s \mL_H} \,  \mL_{\La_j}  \, H_x 
    \\
    &=&
    \dd\Gamma\left(  \i \,  \int_{\R}\dd s \, W_g(s)\, \e^{\i s h} \,  [\lambda_j, h_x]\e^{-\i s h}\right)
    \eqcolon 
    \dd\Gamma((\lambda_j\OD)_x)    
\end{eqnarray*}
and
it was shown in \cite{TW2025} that the sums in the following computation are absolutely convergent in $\mA$:
\begin{eqnarray*}
    [\Lambda_2\OD,\La_1\OD] 
    &\coloneq &
    \sum_{x,y\in\Z^2} [(\Lambda_2\OD)_x,(\Lambda_1\OD)_y]
    \\
    &=&
    \sum_{x,y\in\Z^2}\dd\Gamma([(\lambda_2\OD)_x, (\lambda_1\OD)_y] )
    =
    \dd\Gamma([\lambda_2\OD,\lambda_1\OD])    
\end{eqnarray*}
with 
\begin{equation} \label{lambdajOD}
\lambda_j\OD :=  \i \,  \int_{\R}\dd s \, W_g(s)\, \e^{\i s h} \,  [\lambda_j, h]\e^{-\i s h}\,.
\end{equation}

Now we are left to prove that with
\[
\lambda_j^{\widetilde{\mathrm{OD}}} := P\lambda_jP^\perp + P^\perp \lambda_jP
\]
it holds that 
\[
\tr(P [\lambda_2\OD,\lambda_1\OD])
=\tr(P [\lambda_2\ODT ,\lambda_1\ODT])\,.
\]
But 
since 
\[
\tr(P [\lambda_2\ODT ,\lambda_1\ODT])= \tr(P (\lambda_2\ODT \lambda_1\ODT - \lambda_1\ODT \lambda_2\ODT) P) = \tr(P \lambda_2 P^\perp \lambda_1 P - P \lambda_1 P^\perp \lambda_2 P) \,,
\]
this follows once we can show that
\[
P\lambda_j\OD P^\perp = P\lambda_j P^\perp\,,\quad  P^\perp \lambda_j\OD P= P^\perp\lambda_j P\,,\quad\mbox{and}\quad \tr([P\lambda_2\OD P,P\lambda_1\OD P])= 0 \,.
\]
This was shown  many times for finite dimensional Hilbert spaces and finite rank projections $P$. But it is also easy to see in general: by construction and assumption on $\mu$, there exist closed sets $A,B\subset \R$ with dist$(A,B)>2g$ such that $P:= \chi_A(h)$ and $P^\perp=\chi_B(h)$.
Then,  using the spectral decomposition $h= \int_{\R} u\,\dd P_u $, we find
\begin{eqnarray*}
    P\lambda_j\OD P^\perp  &=& -\i P\int W_g(t) \e^{\i th}[h,\lambda_j] \e^{-\i th}\,\mathrm{d}t P^\perp \\
    &=&\i
    \int W_g(t) \int_\R  \int_\R \e^{\i t(u-v)}(u-v) P\dd P_u \lambda_j \dd P_v\,\mathrm{d}t P^\perp 
    \\
    &=&\i\sqrt{2\pi}
     \int_A  \int_B  \hat W_g(u-v)(u-v) \dd P_u P\lambda_j  P^\perp\dd P_v
     \\
     &=&
    \int_A  \int_B  \dd P_u P\lambda_j  P^\perp\dd P_v
     \\
     &=&    P\lambda_jP^\perp \,,
\end{eqnarray*}
and by taking the adjoint also 
$P^\perp\lambda_j\OD P =P^\perp\lambda_jP $. 

The following two lemmas  show that $P\lambda_2\OD P\lambda_1\OD P$ is trace class, which implies   that
\[
\tr([P\lambda_2\OD P,P\lambda_1\OD P])= 0 \,,
\]
and thus concludes the proof of the proposition.
\end{proof}

\begin{lemma}\label{lem:diagonal-block-strip-local}
Let $h$, $P$, and $\lambda_j$ be as in Proposition~\ref{prop:dcformulas}
and define  $\lambda_j\OD$ as in~\eqref{lambdajOD}. 
Then, for every $N\in\N$, there exists $C_N<\infty$ such that
\[
    \bigl\|(P\lambda_j\OD P)(x,y)\bigr\|
    \le
    C_N
    \langle \|x-y\|\rangle^{-N}
    \langle |x_j|\rangle^{-N}
    \langle |y_j|\rangle^{-N}
\]
for all $x,y\in\Z^2$. Here $\langle r\rangle := 1+r$, and
$\|\cdot\|$ denotes the $\ell^\infty$-norm on $\Z^2$.
\end{lemma}

\begin{proof}
For $j\in\{1,2\}$ the kernel of 
$C_j:=[\lambda_j,h]$  is
\[
    C_j(x,y)
    =
    \bigl(\lambda_j(x)-\lambda_j(y)\bigr)h(x,y).
\]
Since $\lambda_j$ is a half-space switch, $C_j(x,y)$ can be nonzero only
if $x_j$ and $y_j$ lie on opposite sides of the switch. Hence, whenever
$C_j(x,y)\neq0$,
\[
    |x_j|+|y_j|
    \le
    \|x-y\|\, .
\]
Combining this support property with the exponential decay of $h$, we find
that, for every $L\in\N$, there is a constant $C_L<\infty$ such that for all $x,y\in\Z^2$
\begin{equation}\label{eq:Cj-strip-local}
    \|C_j(x,y)\|
    \leq
    C\e^{-c/2\|x-y\|}\e^{-c/2(|x_j|+|y_j|)}
    \leq C_L
    \langle \|x-y\|\rangle^{-L}
    \langle |x_j|\rangle^{-L}
    \langle |y_j|\rangle^{-L}\,.
\end{equation}
 Next we use the following one-body Lieb--Robinson type bound: For any $0<\alpha<c$, where $c$ is the exponent from \eqref{hexp}, set
\[
  c_\alpha
  :=
  \sup_{x\in\Z^2}\sum_{y\in\Z^{2}}
  \|h(x,y)\|\e^{\alpha\|x-y\|_\infty}.
\]
then, by   bounding each term in the exponential series separately, one finds
\begin{equation}\label{nonintLR}
\bigl\|\e^{\i t h}(x,y)\bigr\|\leq \min\left\{1,\e^{c_\alpha|t|- \alpha \|x-y\|_\infty}\right\}\,.
\end{equation}
Estimating \eqref{nonintLR} in the two regimes $\|x-y\|_\infty \leq (1+ 2c_\alpha/\alpha)\langle |t|\rangle $ and 
$\|x-y\|_\infty > (1+ 2c_\alpha/\alpha)\langle |t|\rangle $, one  finds that 
for every $M\in\N$ there exist $Q_M<\infty$ such that, for
all $t\in\R$,
\begin{equation}\label{eq:one-particle-propagator-decay}
    \bigl\|\e^{\i t h}(x,y)\bigr\|
    \le
    Q_M
    \langle |t|\rangle^{M}
    \langle \|x-y\|\rangle^{-M}\,.
\end{equation}
Let
\[
    C_j(t):=\e^{\i t h}C_j\e^{-\i t h}\,,
\]
fix $N\in\N$, choose $R>N+2$, and then choose $M\geq R+N$. Using
\eqref{eq:one-particle-propagator-decay} and
\eqref{eq:Cj-strip-local} with exponent $R$, we obtain
\begin{align*}
    \hspace{2em}&\hspace{-2em}
    \|C_j(t)(x,y)\|
    \\
    &\leq
    \sum_{u,v\in\Z^2}
    \bigl\|\e^{\i t h}(x,u)\bigr\|
    \|C_j(u,v)\|
    \bigl\|\e^{-\i t h}(v,y)\bigr\|                                      \\
    &\le
    C_RQ_M^2\langle |t|\rangle^{2 M}
    \sum_{u,v\in\Z^2}
    \langle \|x-u\|\rangle^{-M}
    \langle \|u-v\|\rangle^{-R}
    \langle \|v-y\|\rangle^{-M}      
    \langle |u_j|\rangle^{-R}
    \langle |v_j|\rangle^{-R}.
\end{align*}
By the reverse triangle inequality,
\[
    1+|x_j|
    \le
    1+|u_j|+\|x-u\|
    \le
    (1+|u_j|)(1+\|x-u\|),
\]
and therefore the Peetre type inequalities
\[
    \langle |u_j|\rangle^{-N}
    \le
    \langle \|x-u\|\rangle^N
    \langle |x_j|\rangle^{-N} 
\qquad\mbox{and}\qquad
    \langle |v_j|\rangle^{-N}
    \le
    \langle \|v-y\|\rangle^N
    \langle |y_j|\rangle^{-N} 
\]
hold.
Using these inequalities and dropping the remaining negative powers that are all smaller or equal to one, we get
\[
    \|C_j(t)(x,y)\|
     \leq
    C_RQ_M^2\langle |t|\rangle^{2 M}
    \langle |x_j|\rangle^{-N}
    \langle |y_j|\rangle^{-N}\!    
    \sum_{u,v\in\Z^2}\!
    \langle \|x-u\|\rangle^{-R}
    \langle \|u-v\|\rangle^{-R}
    \langle \|v-y\|\rangle^{-R} \,.
\]
Since  $R> 2$, the standard convolution estimate on $\Z^2$
gives
\[
    \sum_{u,v\in\Z^2}
    \langle \|x-u\|\rangle^{-R}
    \langle \|u-v\|\rangle^{-R}
    \langle \|v-y\|\rangle^{-R}
    \le
    \widetilde C_R
    \langle \|x-y\|\rangle^{-R}
\]
and thus with $R>N$ also
\begin{equation}\label{eq:Cjt-strip-local}
    \|C_j(t)(x,y)\|
    \le
    C_N
    \langle |t|\rangle^{2 M}
    \langle \|x-y\|\rangle^{-N}
    \langle |x_j|\rangle^{-N}
    \langle |y_j|\rangle^{-N}\,.
\end{equation}
Since $W_g$ is rapidly decaying, the function
$\langle |t|\rangle^{2 M}W_g(t)$ is integrable. Therefore, after
absorbing the $t$-integral into the constant, we obtain
\begin{equation}\label{eq:lambdajOD-strip-local}
    \|\lambda_j\OD(x,y)\|
    \le
    C_N
    \langle \|x-y\|\rangle^{-N}
    \langle |x_j|\rangle^{-N}
    \langle |y_j|\rangle^{-N}.
\end{equation}
Finally, by the spectral gap and exponential localization of $h$, the Fermi
projection has exponentially decaying kernel, and hence for every
$M\in\N$ there exists $C_M>0$ such that
\[
    \|P(x,y)\|
    \le
    C_M
    \langle \|x-y\|\rangle^{-M}.
\]
Using
\[
    (P\lambda_j\OD P)(x,y)
    =
    \sum_{u,v\in\Z^2}
    P(x,u)\lambda_j\OD(u,v)P(v,y),
\]
and applying the same Peetre and convolution argument as above, with
\eqref{eq:lambdajOD-strip-local} in place of
\eqref{eq:Cj-strip-local}, we obtain, after redefining $C_N$,
\[
    \bigl\|(P\lambda_j\OD P)(x,y)\bigr\|
    \le
    C_N
    \langle \|x-y\|\rangle^{-N}
    \langle |x_j|\rangle^{-N}
    \langle |y_j|\rangle^{-N}\,.\qedhere
\]
\end{proof}

\begin{lemma}\label{lem:diagonal-block-products-trace-class}
With the notation of Lemma~\ref{lem:diagonal-block-strip-local}, the operator
\[
    (P\lambda_1\OD P)(P\lambda_2\OD P)
    =
    P\lambda_1\OD P\lambda_2\OD P
\]
and thus also its adjoint
\[
    (P\lambda_2\OD P)(P\lambda_1\OD P)
    =
    P\lambda_2\OD P\lambda_1\OD P
\]
is trace class.
\end{lemma}

\begin{proof}
For $j=1,2$ set
    $D_j:=P\lambda_j\OD P$.
By Lemma~\ref{lem:diagonal-block-strip-local}, for every $N\in\N$ there is
$C_N<\infty$ such that
\[
    \|D_1(x,z)\|
    \le
    C_N
    \langle \|x-z\|\rangle^{-N}
    \langle |x_1|\rangle^{-N}
    \langle |z_1|\rangle^{-N}
\]
and
\[ 
    \|D_2(z,y)\|
    \le
    C_N
    \langle \|z-y\|\rangle^{-N}
    \langle |z_2|\rangle^{-N}
    \langle |y_2|\rangle^{-N}\,.
\]
Choose $N>2$. Then
\begin{align*}
    \hspace{2em}&\hspace{-2em}
    \sum_{x,y\in\Z^2}\|(D_1D_2)(x,y)\|
    \\
    &\leq
    C_N^2
    \sum_{x,y,z\in\Z^2}
    \langle \|x-z\|\rangle^{-N}
    \langle |x_1|\rangle^{-N}
    \langle |z_1|\rangle^{-N}    
    \langle \|z-y\|\rangle^{-N}
    \langle |z_2|\rangle^{-N}
    \langle |y_2|\rangle^{-N}\,.
\end{align*}
For fixed $z$, since $N>2$,
\[
    \sum_{x\in\Z^2}
    \langle \|x-z\|\rangle^{-N}
    \langle |x_1|\rangle^{-N}
    \le
    \sum_{x\in\Z^2}
    \langle \|x-z\|\rangle^{-N}
    \le C_N'\,,
\]
uniformly in $z$. Similarly,
\[
    \sum_{y\in\Z^2}
    \langle \|z-y\|\rangle^{-N}
    \langle |y_2|\rangle^{-N}
    \le C_N'\,,
\]
uniformly in $z$. Therefore
\[
    \sum_{x,y\in\Z^2}\|(D_1D_2)(x,y)\|
    \leq
    (C_N')^2
    \sum_{z\in\Z^2}
    \langle |z_1|\rangle^{-N}
    \langle |z_2|\rangle^{-N}
    <\infty\,.
\]
Thus $D_1D_2$ has an absolutely summable matrix kernel. Since the fiber
dimension is finite, this implies that $D_1D_2$ is trace class. 
\end{proof}

\appendix

\section{Basic notions on fermionic lattice systems}\label{sec:basic}

We basically use the standard formalism for describing infinitely extended systems of interacting lattice fermions  in the thermodynamic limit, as explained for example in \cite{bratteliI,bratteliII},  together with more specific results for fermionic systems from~\cite{ArakiMoriya2002}. 
However, it turns out that our analysis requires a number of  additional  technical results, mostly concerning periodic interactions and position operators, which were not available in the previous literature. 
The proofs of the corresponding lemmas and propositions  are collected in \cite[Appendix~B]{WMMMT24}.

\subsection{(Quasi-)local fermionic interactions}

 The anti-symmetric  (or fermionic) Fock space over the lattice $\Z^d$ is given by
    \begin{align*}
       \mathcal{F}(\Z^d,\C^\fd) := \bigoplus_{N=0}^{\infty}\ell^2(\Z^d,\C^\fd)^{\wedge N} ~,
    \end{align*}
where $\ell^2(\mathbb{Z}^d, \mathbb{C}^\fd)^{\wedge N}$ denotes the $N$-fold anti-symmetric tensor product of $\ell^2(\mathbb{Z}^d, \mathbb{C}^\fd)$, with the convention $\ell^2(\mathbb{Z}^d, \mathbb{C}^\fd)^{\wedge 0} := \mathbb{C}$.
    We use $a^*_{x,i}$ and $a_{x,i}$, for $x\in \Z^d$ and $i\in \{1,\dots,\fd\}$, to denote the fermionic creation and annihilation operators associated to the standard basis of $\ell^2(\Z^d,\C^\fd)$; recall that they satisfy the canonical anti-commutation relations (CAR).
    The number operator at site $x\in\Z^d$ is defined by
    \begin{align*}
        n_x:=\sum_{i=1}^\fd a^*_{x,i}a_{x,i}\,.
    \end{align*}
    The algebra of all bounded operators on $\mathcal{F}(\Z^d,\C^\fd)$ is denoted by $ \mathcal{L}(\mathcal{F}(\Z^d,\C^\fd))$. 
    For each $M\subseteq \Z^d$ let $\mA_M$ be the C*-subalgebra of $ \mathcal{L}(\mathcal{F}(\Z^d,\C^\fd))$  generated by
    \begin{align*}
        \{a^*_{x,i}~|~x\in M,~ i\in \{1,\dots,\fd\}\}~.
    \end{align*}
    The C*-algebra $\mA := \mA_{\Z^d}$ is the CAR-algebra, which we also call the quasi-local algebra. We define $P_0(\Z^d) := \{M\subseteq \Z^d~|~|M|<\infty \}$ and call
    \begin{align*}
        \mA_0 := \bigcup_{M\in P_0(\Z^d)} \mA_M \subseteq \mA
    \end{align*}
    the local algebra, which is dense in $\mA$ with respect to the norm topology. Consequently, an operator   in $ \mathcal{L}(\mathcal{F}(\Z^d,\C^\fd))$ is called quasi-local if it lies in $\mA$ and local if it lies in $\mA_0$.
    For each $\varphi \in \R$ there is a unique $*$-automorphism $g_\varphi$ of $\mA$, such that
    \begin{align*}
        g_\varphi(a^*_{x,i}) =  \e^{\i\varphi} \, a^*_{x,i}, \quad \text{for all}~~ x\in \Z^d,~ i\in \{ 1,\dots,\fd \} ~.
    \end{align*}
One defines the set
    \begin{align*}
        \mA^N := \{A\in \mA~|~ g_\varphi (A) = A \; \forall \varphi \in \R\}
    \end{align*}
    and calls $\mA^N$ the gauge-invariant sub-algebra of $\mA$. It is the closure of the set of all local observables that commute with the number operator, i.e.\  all local observables $A\in   \mA_0$ that satisfy $A\in \mA_M \Rightarrow [A,N_M]:=[A,\sum_{x\in M}n_x]=0$.  Its part in $M\subseteq\Z^d$  is denoted by $\mA^N_M:= \mA^N\cap \mA_M$.
For disjoint regions $M_1,M_2 \subseteq\Z^d$,  $M_1\cap M_2 = \emptyset$, operators  $A\in \mA_{M_1}^N$ and $B\in \mA_{M_2}$ commute:  $[A,B] = 0$.

In order to define quantitative notions of localization for quasi-local operators, one makes use of the fact that one can localize operators to given regions by means of the fermionic conditional expectation. 
To this end first note that  $\mA$ has a unique state $\w^{\tr}$ (commonly referred to as a tracial state) that satisfies
    \begin{align*}
        \w^{\tr}(AB) = \w^{\tr}(BA)
    \end{align*}
    for all  $A,B \in \mA$ (e.g.\ \cite[Definition 4.1, Remark 2]{ArakiMoriya2002}).

\medskip

\begin{proposition}[{\cite[Theorem 4.7]{ArakiMoriya2002},\cite[Proposition 2.1]{WMMMT24} }]\label{Ex+UniqueExpectation}
    For each $M\subseteq\Z^d$ there exists a unique linear map 
    \begin{align*}
        \E_M:\mA \to \mA_M\,,
    \end{align*}
    called the conditional expectation with respect to $\w^\tr$, such that
    \begin{equation}\label{eq:Conditional expectation defining property}
      \forall A\in \mA \; \;\forall B\in \mA_M \,:\quad \w^{\tr}(AB)=\w^{\tr}(\E_M(A)B) \,.
    \end{equation}
    It is unital, positive and has the properties 
    \begin{eqnarray*}
        \forall M\subseteq \Z^d\;\; \forall A,C\in \mA_M\;\; \forall B\in\mA\, : &&   \E_M \br{A\,B\,C} = A\, \E_M(B)\,C\\[1mm]
      \forall M_1,M_2 \subseteq \Z^d\,:&&   \E_{M_1} \circ \E_{M_2}  = \E_{M_1\cap M_2}\,\\[1mm]
        \forall M \subseteq \Z^d\,:&& \E_M \mA^N \subseteq \mA^N
  \\[1mm]
        \forall M \subseteq \Z^d \; \;  \forall A \in \mA\,:&&  \norm{\E_M(A)}\leq \norm{A}\,.
    \end{eqnarray*}
\end{proposition}

With the help of $\E$ we can define subspaces of $\mA^N$ that contain operators with well-defined decay properties (cf.\ \cite{MO20} for similar definitions for quantum spin systems).

\medskip

\begin{definition} \label{dfn:Dinfty}
    For $\nu \in \N_0$ and $A \in \mA$ let
    \begin{align*}
        \norm{A}_\nu := \norm{A} + \sup_{k\in \N_0}( \norm{A-\E_{\La_k}A} (1+k)^\nu) 
    \end{align*}
    where $\La_k = \{x\in \Z^d \, | \, \norm{x}_\infty \leq k \}$ is the box with side-length $2k+1$ around $0$ ($\norm{\cdot}_\infty$ is the maximum norm on $\Z^d$).
    We denote the set of all $A\in \mA$ with finite $\norm{\cdot}_\nu$ by $\mA_\nu$ and also define $\mA_{\infty} := \bigcap_{\nu \in \N_0} \mA_\nu$. We will sometimes refer to these norms as \emph{decay norms}. We also use the notation $\mA_\nu^N := \mA_\nu \cap \mA^N$ and $\mA_\infty^N := \mA_\infty \cap \mA^N$. Finally, we say that $A\in \mA_{\exp}$ iff $A\in\mA$ and there is a $a>0$ such that $\sup_{k\in \N_0}( \norm{A-\E_{\La_k}A} \e^{ak}) <\infty$.
\end{definition}

The relevant physical dynamics on $\mA$ is generated by densely defined derivations, which in turn are constructed from so-called interactions.
    An \emph{interaction} is a map $\Phi: P_0(\Z^d) \to \mA^N$, such that $\Phi(\emptyset) = 0$ and for all $M\in P_0(\Z^d)$ it holds that $\Phi(M) \in \mA_M$, $\Phi(M)^* = \Phi(M)$, and the sum 
    \[\sum_{\substack{K\in P_0(\Z^d)\\ M\cap K \neq \emptyset}} \Phi(K)\] 
    converges unconditionally, meaning that the partial sums converge to the same value for all sequences that enumerate the index set. 
    For two interactions   $\Phi$ and $\Psi$ their commutator is given by
    \begin{align*}
        [\Phi,\Psi]: P_0(\Z^d) \to \mA^N\,,\qquad M\mapsto 
        [\Phi,\Psi](M) \;:= \sum_{\substack{M_1,M_2 \subseteq M\\ M_1 \cup M_2= M}} [\Phi(M_1),\Psi(M_2)] \, .
    \end{align*}
    The map $\i[\Phi,\Psi]$ satisfies the definition of an interaction except for the last requirement, which is not always satisfied. All the commutators of interactions appearing in the following will however be interactions (see  Lemma~\ref{lem: sum representation of generator} and Proposition~\ref{commutator of interactions}).

Interactions  define derivations on the algebra in the following way.    
    For an interaction~$\Phi$ let
    \begin{align*}
        \mL_{\Phi}^\circ: \mA_0 \to \mA,\quad A\mapsto \sum_{M\in P_0(\Z^d)} [\Phi(M),A] \, .
    \end{align*}
    It follows from \cite[Propositions~3.1.15 and~3.2.22]{bratteliI} that $\mL_{\Phi}^\circ$ is closable. We denote its closure by $\mL_{\Phi}$ and call it the Liouvillian of $\Phi$.
Some basic interactions that will appear frequently in the following are the number operator $N$ and the position operators $X_j$ for $j\in\{1,\dots, d\}$, which are non-vanishing only on one-element sets. They are defined by $N(\{x\}) = n_x$ and $X_j(\{x\})= x_j n_x$ for $x\in\Z^d$.

In the operator-algebraic framework, states are described as positive linear functionals $\w \colon \mA \to \C$ of norm $1$. Given an interaction $\Phi$, a state $\w$ is called a $\mL_{\Phi}$-ground state, or simply ground state of $\Phi$, if it holds that 
\begin{align*}
 \forall A \in D(\mL_\Phi) : \quad  \w(A^* \mL_\Phi A) \geq 0   \, .
\end{align*}
Here, for a given linear operator $ T $ acting on $\mathcal{A}$, $ D(T) $ denotes its domain.
For a state $\w$ on $\mA$ and an interaction $\Phi$ we define 
   \begin{equation} \label{eq:omegaPUV}
   \overline{\w}(\Phi) := \lim_{k\to \infty} \frac{1}{|\La_k|}\sum_{M\subseteq \La_k} \w(\Phi(M))\, ,
   \end{equation}
   whenever the limit exists  and call it the per-volume expectation value of $\Phi$ with respect to $\w$.

\medskip

\begin{definition}\label{norm interaction}
    Let $\Phi$ be an interaction and $\nu \in \N_{0}$. Let
    \begin{align*}
        \|\Phi\|_{\nu} := \sup_{x\in \Z^d} \sum_{\substack{M\in
        P_0(\Z^d)\\x\in M}}(1+\mathrm{diam}(M))^\nu  \| \Phi(M)\|  ~.
    \end{align*}
    The set of all interactions with finite $\norm{\cdot}_{\nu}$ is denoted by $B_{\nu}$. We also define $B_{\infty} := \bigcap_{\nu \in \N_0} B_{\nu}$. \\
    For $a>0$ let
    \begin{align*}
        \|\Phi\|_{\exp,a} := \sup_{x\in \Z^d} \sum_{\substack{M\in 
        P_0(\Z^d)\\x\in M}} \exp(a \, \mathrm{diam}(M))  \| \Phi(M)\| 
    \end{align*}
    and denote the set of all interactions with finite $\norm{\cdot}_{\exp,a}$ for some $a>0$ by $B_{\exp}$. By $\mathrm{diam}(M)$ we mean the maximal distance of two elements in $M\subset\Z^d$ with respect to~$\norm{\cdot}_\infty$.
\end{definition}

\medskip

\begin{lemma}\label{lem: sum representation of generator}
    For a $B_\infty$-interaction $\Phi$, $j\in \{ 1,\dots, d \}$ and $p,q\in \R$ it holds that $\mA_\infty \subseteq D(\mL_{p \Phi + q X_j})$. For $ A\in \mA_\infty$, the sums 
    \[\sum_{M\in P_0(\Z^d)} [\,\Phi(M),A\,] \quad \text{and} \quad \sum_{x\in \Z^d} [\,x_j\,n_x,A\,]\] converge absolutely and  
    \[ \mL_{p \Phi + q X_j}A = p \sum_{M\in P_0(\Z^d)} [\,\Phi(M),A\,] + q \sum_{x\in \Z^d} [\,x_j\,n_x,A\,] \in \mA_\infty \,. \]
    For each $\nu \in \N_0$ there is a constant $c_\nu$, independent of $\Phi, j, p, q, A$ such that
    \[ \norm{ \mL_{p \Phi + q X_j}A }_\nu \leq c_\nu \, (|p| \,\norm{\Phi}_{d+1+2\nu} + |q| \,\norm{n_0}) \, \norm{A}_{d+3+2\nu}\,.\]
\end{lemma}

A cocycle of automorphisms is a family $(\alpha_{u,v})_{(u,v) \in \R^2}$ of automorphisms of $\mA$ such that $\alpha_{u,u} = \mathrm{id}$ for all $u\in \R$ and $\alpha_{u,v}\circ \alpha_{v,w} = \alpha_{u,w} $ for all $u,v,w\in \R$.
We say a cocycle of automorphisms is generated by the family of interactions $(\Phi^{v})_{v\in \R}$ if for all $A \in \mA_0$ it holds that $ \partial_v \alpha_{u,v} A = \alpha_{u,v} \i \mL_{\Phi^{v}} A$. In the case where the generating family is a constant interaction $\Phi$, the family $(\alpha_{0,v})_{v\in \R}$ is a one-parameter group of automorphisms with generator $\i\,\mL_{\Phi}$. One can show that the family of all gauge automorphisms $(g_\varphi)_{ \varphi \in \R}$ is generated by the number operator $N$ in this sense and that, for an observable $A\in D(\mL_N)$, gauge-invariance ($A\in \mA^N$) is equivalent to $\mL_N  A = 0$.

\medskip

\begin{lemma}\label{lem:cocycles}
    If $(\Phi^{v})_{v\in \R}$ is a family of interactions such that $\sup_{v\in \R} \norm{\Phi^{v}}_\nu < \infty $ for all $\nu\in \N_0$ and for each $M \in P_0(\Z^d)$ the map $ v \mapsto \Phi^{v}(M)$ is continuous with respect to the norm topology in $\mA$, then there is a unique cocycle of automorphisms $(\alpha_{u,v})_{(u,v) \in \R^2}$ generated by $(\Phi^{v})_{v\in \R}$.
    This cocycle has the property that for all $\nu\in \N_0$ there is a constant $c_\nu$ such that 
    \begin{align*}
        \forall A\in \mA_\infty,\, u,v\in \R\colon \quad\norm{\alpha_{u,v} A }_\nu \leq c_\nu \exp(c_\nu\, \sup_{t\in\R}\norm{\Phi^t}_{2d+1+\nu} \, |u-v|) \norm{A}_\nu   \, .
    \end{align*}
    It also holds for all $u,v\in \R$ that $\alpha_{u,v} \mA_\infty = \mA_\infty$.
    If additionally $\sup_{v\in \R} \norm{\Phi^{v}}_{\exp,a} < \infty $ for some $a > 0$, then for all $\nu \in \N_0$ there exists a continuous function $b_\nu:\R \to \R$, growing at most polynomially at infinity, such that
    \begin{align*}
        \norm{\alpha_{u,v} A}_\nu \leq b_\nu(|u-v|) \norm{A}_\nu \quad \text{ for all } A\in \mA_\infty,\, u,v\in \R \,.
    \end{align*}
\end{lemma}

\subsection{Periodic interactions}

From now on we will specialise to periodic interactions and periodic states for which the existence of the per-volume expectation is easy to see.  

\medskip

\begin{definition}\label{def:translation}
We denote by $\mathrm{Aut}(\mA)$ the group of automorphisms on $\mA$.
    A \emph{translation} is a map $T:\Z^d \to \mathrm{Aut}(\mA)$ that assigns to each possible shift vector $\g$ an automorphism $T_\g$ and satisfies the following properties:
    \begin{itemize}
        \item[(i)] For all $\g\in \Z^d$ and $M \subseteq \Z^d$ it holds that $T_{\g}(\mA_M) = \mA_{M+\g}$.
        \item[(ii)] For all $\g\in \Z^d$ and $x\in \Z^d$, it holds that $T_{\g}n_x = n_{x+\g}$.
    \end{itemize}
\end{definition}

\begin{remark}
    Note that the definition of a translation does \emph{not} require it be a group homomorphism.  As a consequence,
   the  magnetic translations defined in \eqref{magtrans} are translations in the above sense for \emph{all} $b\in\R$.
   Translations do however satisfy a homomorphism property when acting on certain elements like the local number operators or terms of interactions that are periodic with respect to the translation. 
\end{remark}

  The two properties of the definition guarantee that translations are compatible with the conditional expectation and gauge automorphisms respectively:

\medskip
\begin{lemma}\label{lem: compatibility of translations}
    Let $T$ be a translation and $\g \in \Z^d$. It holds that 
    \begin{align*}
        \forall M \subseteq \Z^d \colon& \quad T_\g \, \E_M   =  \E_{M + \g} \, T_\g \, ,\\
        \forall \varphi \in \R\colon& \quad g_\varphi \, T_\g = T_\g \, g_\varphi \, .
    \end{align*} 
\end{lemma}

    Let $T$ be a translation. We say that a state $\w$ is $T$-periodic if
\begin{align*}
\forall \, \gamma \in \Z^d, \, A\in\mathcal{A}:\quad
    \w(T_\gamma A) = \w(A)\, .
\end{align*}
    An interaction $\Phi$ is called $T$-periodic if
    \begin{align*}
\forall \, \gamma \in \Z^d, \, M\in P_0(\Z^d):\quad        T_{\g} \Phi(M) = \Phi(M+\g)\,. 
    \end{align*}

\medskip

In order to define a standard representation of a periodic interaction in terms of a quasi-local observable, we introduce the following notions. 
    We say $M_1,M_2 \in P_0(\Z^d)$ have the same shape if there is a $\gamma\in \Z^d$ such that $M_1 = M_2 +\g$. For each $M\in P_0(\Z^d)$ we define $s(M) \in \Z^d$ as the vector that shifts the center of mass of $M$ closest to $0$. More precisely it is defined as the unique element of $\Z^d$ such that
    \begin{align}\label{def: shape}
         \sum_{x\in \, M + s(M)} \frac{x}{\abs{M}} \in  {\big ( }-\tfrac{1}{2}\,,\,\tfrac{1}{2}\, {\big]}^d \, .
    \end{align}
    We call $M+s(M)$ the standard representative of $M$ at $0$ and denote the set of all standard representatives at $0$ by $R_0(\Z^d)$.
We call an interaction $\Phi$ absolutely summable at $0$ if 
    \begin{align*}
        \sum\limits_{\substack{M\in R_0(\Z^d)}}  \Phi(M)
    \end{align*}
    converges absolutely. In that case   we denote the limit by $\Phi_0$.

The following construction shows that, under an additional compatibility condition, adding up all the translates of a quasi-local observable yields a $T$-invariant interaction.

\medskip

\begin{definition}
    Let $T$ be a translation. Then $A\in \mA_\infty$ is called \emph{$T$-compatible} if $T_\g\, T_{\tilde\g} \,A = T_{\g+\tilde\g}\, A $ for all $\g, \tilde\g \in \Z^d $.
\end{definition}

\medskip

\begin{proposition}\label{prop: interaction associated to observable}
    Let $T$ be a translation and $A\in \mA_\infty^N$ be self-adjoint and $T$-compatible. Then the map $\Phi_{A}^T\colon P_0(\Z^d) \to \mA^N  ,\; M\mapsto \Phi_{A}^T(M)$ with
    \begin{align*}
        \Phi_{A}^T(M)\coloneq
        \left\{\begin{array}{cl}
        T_{\g}\E_{\La_{0}}A & \mbox{if  $M=\La_0 +\g$ for some $\g\in\Z^d$}\\
        T_{\g}(\E_{\La_{k}}A - \E_{\La_{k-1}}A)&
        \mbox{if  $M=\La_k +\g$ for some $\g\in\Z^d$ and $k\in\N$}\\
        0 & \mbox{otherwise}
        \end{array}
        \right
        .
    \end{align*}
   defines a  $T$-periodic $B_\infty$-interaction with $(\Phi_A^T)_0 = A$. Conversely, if $\Phi$ is a $T$-periodic $B_\infty$-interaction, then it is absolutely summable at $0$ and $\Phi_0$ lies in $\mA_\infty$, is self-adjoint, and $T$-compatible.
\end{proposition}

\medskip

\begin{proposition}\label{commutator of interactions}
    Let $\Phi$ and $\Psi$ be $B_\infty$-interactions and $\omega$ a state such that $\Phi$, $\Psi$, and $\w$ are all  $T$-periodic with respect to a translation $T$. Moreover, let $j \in \{1,\dots,d\}$, $p,q \in \R$.
    \begin{itemize}
        \item[\rm(i)] $\i[p\Phi + qX_j, \Psi]$ is a $T$-periodic $B_\infty$-interaction.
        \item[\rm(ii)] $\i[X_j,\Psi]_0  = \i\mathcal{L}_{X_j}(\Psi_0)$.
        \item[\rm(iii)] $\mL_{\i[\Psi,p\Phi+qX_j]}A = \sum_{\mu \in \Z^d} [\i \, \mL_\Psi \, (p\, T_\mu\, \Phi_0 + q \, \mu_j \,  n_\mu) ,\, A\, ]$ for all $A\in \mA_\infty$.
        \item[\rm(iv)] $\omega(\i[\Phi,\Psi]_0)  =\omega(\i\mathcal{L}_\Phi(\Psi_0)) =-\omega(\i\mathcal{L}_\Psi(\Phi_0))$.
    \end{itemize}
\end{proposition}

\subsection{The off-diagonal map}\label{sec:offdiag}

We need to introduce a further technical ingredient, namely the map that selects the off-diagonal part of an operator or of an interaction with respect to a ground state~$\w_0$, as well as the quasi-local inverse of the Liouvillian associated to the Hamiltonian. The former  map was  introduced in \cite{hastings2005quasiadiabatic,BMNS12} for finite gapped systems and extended in \cite{MO20} to infinite systems with a gap in the bulk. Here we further generalize it by composing it with automorphisms of the algebra $\mA$.
 The definition of both objects involves the choice of a function 
$W_g: \R\to \R$  having the following properties for some $g>0$:
    \begin{itemize}
        \item[(i)] $W_g$ is odd;
        \item[(ii)] the Fourier transform $\widehat W_g$ of $W_g$ satisfies $\widehat{W_g}(k) = \frac{- \i}{\sqrt{2  \pi}\, k}$ for  $k \in \R \setminus [-g,g]$;
        \item[(iii)] $\sup_{s\in \R} \bigl( \lvert s \rvert ^n \lvert W_g(s) \rvert \bigr) < \infty$ for all $n\in\N_0$. 
    \end{itemize}
 It is easy to see that such a function exists for any $g>0$, and an explicit example with additional properties is given in \cite[Lemma 2.6]{BMNS12} (see also \cite[Lemma 2.3 and Equation 2.12]{BMNS12}).

\medskip

\begin{lemma}\label{lem: almost local obs. for OD and inv. liou.}
    Let $T$ be a translation, $H\in B_\infty$ a $T$-periodic interaction, $\alpha\in \mathrm{Aut}(\mA)$ be such that $\alpha \mA_\infty = \mA_\infty $ and $\alpha T_\g = T_\g \alpha$ for all $\g \in \Z^d$, and $\Psi$ an interaction of the form $\Psi= p\, Y + q\, X_j$, where $p,q \in \R$, $j\in\{1,\dots,d\}$ and $Y \in B_\infty$ is   $T$-periodic. Then the operators 
    \begin{align*}
        (\Psi\OD[\alpha])_* :=  \alpha^{-1} \int_{\R}\mathrm{d}s \, W_g(s) \, \e^{\i s\mL_{H}} \,  \alpha \,  \i \,\mL_\Psi \, \alpha^{-1} \, H_0
    \end{align*}
    and
    \begin{align*}
        \mathcal{I}(\Psi)_* \coloneq  - \int_{\R}\mathrm{d}s \,W_g(s) \int_{0}^s \mathrm{d}u  \, \e^{\i u\mL_{H}} \, \i \, \mL_\Psi \, H_0 \, ,
    \end{align*}
    defined  via Bochner integrals in  $\mathcal{A}$,
    are in $\mA_\infty$, self-adjoint, and $T$-compatible.
\end{lemma}

\begin{definition}\label{def:off-diagonal}
Given the assumptions of Lemma \ref{lem: almost local obs. for OD and inv. liou.}, we define the $T$-periodic $B_\infty$-interactions
    \begin{equation}\label{dfn:XjOD}
    \Psi\OD[\alpha] \coloneq \Phi_{(\Psi\OD[\alpha])_*}^T\quad\mbox{and}\quad \mathcal{I}(\Psi) \coloneq \Phi_{\mathcal{I}(\Psi)_*}^T\,
    \end{equation}
using Proposition \ref{prop: interaction associated to observable}.
For $\alpha=\mathrm{id}$,we denote $\Psi\OD[\alpha]$ simply by $\Psi\OD$.
Note that the interactions $\Psi\OD[\alpha]$ and $\mathcal{I}(\Psi)$ are defined in such a way that ${(\Psi\OD[\alpha])_0} = {(\Psi\OD[\alpha])_*}$ and $\mathcal{I}(\Psi)_0 = \mathcal{I}(\Psi)_*$. From now on, we will always refer to these quasi-local operators as ${(\Psi\OD[\alpha])_0} $ and $\mathcal{I}(\Psi)_0$.
\end{definition}

\medskip

We call the map $\Psi\mapsto  \Psi\OD[\alpha]$ the off-diagonal map, and  refer to $\Psi\OD[\alpha]$ as the off-diagonal part of $\Psi$ (the dependence on $\alpha$ is left implicit if no ambiguity can arise).

\section{Some technical tools} \label{appendix_technical}

We collect here a number of results that are instrumental to the proofs of the main results in the paper.

\subsection{Basics on non-interacting periodic magnetic systems}\label{app:BF}

We review here the Bloch--Floquet decomposition of operators commuting with a unitary representation of a lattice of (magnetic) translations, as well as their finite-volume approximations with periodic boundary conditions (the so-called supercell restrictions: cf.\ e.g.~\cite[Sec.~4]{gontier2016convergence}).

The non-interacting  Hamiltonians $\mathfrak{h}^b$ acting on the one-body space $ \ell^2(\Z^2,\C^\fd)$ were defined in Section~\ref{sec:models}. 
We check \eqref{eqn:hb_tb_inv}:
\begin{eqnarray*}
    ((t^b_\gamma)^{-1}\mathfrak{h}^bt^b_\gamma  f)(x) &=& \sum_{y'} \e^{\i b (x_1-\gamma_1)\gamma_2}\mathfrak{h}^b(x-\gamma,y') \e^{-\i b y_1'\gamma_2} f(y'+\gamma)\\
    &=& \sum_{y=y'+\gamma} \e^{\i b (x_1-\gamma_1)\gamma_2} \mathfrak{h}^b(x-\gamma,y- \gamma) \e^{-\i b (y_1-\gamma_1)\gamma_2} f(y )\\
    &=&\sum_y \e^{\i b x_1\gamma_2} \mathrm{e}^{ \mathrm{i}\frac{x_2+y_2-2\gamma_2}{2} b (x_1-y_1)}  h(x-y) \e^{-\i b  y_1 \gamma_2} f(y )\\
    &=&\sum_y   \mathrm{e}^{ \mathrm{i}\frac{x_2+y_2}{2} b (x_1-y_1)}  h(x-y)   f(y )\\
    &=& (\mathfrak{h}^bf)(x)\,.
\end{eqnarray*}

For $b=2\pi \frac{p}{q} \in Q$ 
magnetic translations by vectors in    the sublattice $\Z^2_q:= \Z\times (q \Z)\subset \Z^2$ are actually   standard translations:
\[
(t^b_\gamma f)(x)=  \e^{-\i b x_1\gamma_2} f(x+\gamma) = f(x+\gamma) = (t^0_\gamma f)(x) \quad \forall\gamma\in \Z^2_q\,,
\]
as the argument of the magnetic phase is an integer multiple of $2\pi$. This allows us to use the Bloch--Floquet transformation with respect to this sublattice in order to diagonalize $\Z^2_q$-periodic operators.

To define it,  we decompose $x\in\Z^2$ as $x=x_q + [x]_q$ with $x_q\in\Z^2_q$ and $[x]_{q}\in \{(0,0),\ldots,(0,q-1)\}$, and define $I_q:= \{0,\ldots,q-1\}$.
This allows us to  identify 
\[
\iota_q: \ell^2(\Z^2,\C^\fd)\to \ell^2(\Z^2_q \times I_q, \C^{\fd})\,,\quad f\mapsto (\iota_q f)(x,j) := f(x+(0,j))
\]
and map this space to functions on the magnetic Brillouin zone $\T^2_q:= [0,2\pi)\times [0,2\pi/q)$ equipped with normalized Lebesgue measure $\mathrm{d}\mu_q := \frac{q\,\mathrm{d}\lambda^2}{(2\pi)^2}$ :
\[
\mathcal{F}_q: \ell^2(\Z^2;\C^{\fd}) \to L^2(\T^2_q \times I_q,\mathrm{d}\mu_q \times \nu;\C^{\fd})\,,\quad f\mapsto \hat f_j(k) := \sum_{x\in\Z^2_q}\e^{\i k x} \e^{\i k_2 j} f(x+(0,j))\,,
\]
with inverse
\[
\mathcal{F}^{-1}_q:     L^2(\T^2_q \times I_q,\mathrm{d}\mu_q \times \nu;\C^{\fd})\to \ell^2(\Z^2 ;\C^{\fd})\,,\quad f \mapsto \check f(x) := \int_{\T^2_q} \e^{-\i k x} f_{[x]_q}(k)\, \mathrm{d}\mu_q(k)\,.
\]
Let $A\in\mL(\ell^2(\Z^2,\C^\fd))$ be any bounded operator that is translation invariant with respect to translations in $\Z^2_q$, i.e., such that its kernel satisfies 
\[
A(x+\gamma,y+\gamma) = A(x,y) \quad \forall \gamma\in\Z^2_q\,,\;x,y\in \Z^2\,,
\]
and is thus of the form 
\[
A(x_q+[x]_q, y_q+[y]_q) = A(x_q-y_q + [x]_q,[y]_q) =: a_{[x]_q,[y]_q} (x_q-y_q)\,.
\]
Such   operators are matrix-valued multiplication operators in the Bloch--Floquet representation:
\begin{align*}
    \hspace{2em}&\hspace{-2em}
    (\mathcal{F}_q\,A\,\mathcal{F}^{-1}_q \hat f)_j(k)
    \\
    &=
    \sum_{x\in\Z^2_q}\e^{\i k x} \e^{\i k_2 j} \sum_{y\in\Z^2} A(x+(0,j),y) \int_{\T^2_q} \e^{-\i t y} \hat f_{[y]_q}(t) \,\mathrm{d}\mu_q(t)
    \\
    &=
    \sum_{x\in\Z^2_q}\e^{\i k x} \e^{\i k_2 j} \sum_{y\in\Z^2_q} \sum_{l\in I_q} \underbrace{A(x+(0,j),y+(0,l))}_{\eqcolon  a_{jl}(x-y)} \int_{\T^2_q} \e^{-\i t (y+(0,l)) } \hat f_l(t) \,\mathrm{d}\mu_q(t)
    \\
    &=
    \int_{\T^2_q} \sum_{z\in\Z^2_q} 
    \underbrace{\sum_{x\in\Z^2_q}\e^{\i  x(k-t)}}_{=\delta(k-t)} \e^{\i (k_2 j-t_2l)} \sum_{l\in I_q}   a_{jl}(z)\e^{ \i t  z } \hat f_l(t) \,\mathrm{d}\mu_q(t)
    \\
    &=
    \sum_{l\in I_q}  \sum_{z\in\Z^2_q}    a_{jl}(z)\e^{ \i k  (z+(0,j-l)) } \hat f_l(k) \\
    &\eqcolon
    \sum_{l\in I_q}  \hat a_{jl}(k)\,\hat f_l(k)\,.    
\end{align*}
In particular, this applies to our Hamiltonian, for which we have
\[ \mathfrak{h}^b(x+(0,j),y+(0,l)) = \mathrm{e}^{ \mathrm{i}\frac{x_2+y_2+ j+l}{2} b (x_1-y_1)} h(x-y+(0,j-l)) =: h^b_{jl}(x-y)\]
and thus 
\[
(\mathcal{F}_q\mathfrak{h}^b\mathcal{F}^{-1}_q \hat f)(k)  = \hat h^b(k) \,\hat f (k)\quad\mbox{with}\quad 
\hat h^b_{jl}(k) := 
    \sum_{z\in\Z^2_q}    h^b_{jl}(z)\e^{ \i k  (z+(0,j-l)) }  \,.
\]
Exponential decay of $h(z)$ implies that the above Bloch--Floquet fibers $\hat h^b_{jl}(k)$ depend analytically on $k \in \T^2_q$, i.e.\ they can be extended to a complex-analytic function defined on a strip around the ``real axis'' in the space of complex quasi-momenta. This analyticity is inherited from the fibers of the Fermi projection $P^{\mathfrak{h}^b}_{(-\infty,\mu)}$  via the Riesz formula, provided $\mu$ lies in a spectral gap of $\mathfrak{h}^b$; in turn, this yields that the integral kernel $P^{\mathfrak{h}^b}_{(-\infty,\mu)}(x,y)$ decays exponentially away from the diagonal, compare \eqref{eqn:decay_of_kernel}.

Next we discuss the finite volume restrictions
of $\Z^2_q$-periodic operators to boxes $\Lambda_{k_n}$ that are commensurate with the $\Z^2_q$-periodicity: let 
$k_n:= \frac{1}{2}((2n+1)q-1)\in\N$ and $L:= (2n+1)q$.
We decompose $x\in\Z^2$ as $x=x^n +[x]^n$ with $x^n\in(L\Z)^2\subset \Z^2_q$ and $[x]^n\in\Lambda_{k_n}$ and define the restriction of  a $\Z^2_q$-periodic operator $A$   on $\ell^2(\Z^2,\C^\fd)$   to $\ell^2(\Lambda_{k_n},\C^\fd)$ to be the operator $A_n$ with kernel
\[
A_n(x,y) := \sum_{z\in (L\Z)^2} A(x,y+z)\,.
\]
  The (dual) magnetic translations on $\ell^2(\Lambda_{k_n},\C^\fd)$ are   defined by
\[
t^b_{n,\gamma}:\ell^2(\Lambda_{k_n},\C^\fd)\to \ell^2(\Lambda_{k_n},\C^\fd)\,,\quad f \mapsto (t^b_{n,\gamma} f)(x) := \e^{-\i b x_1\gamma_2}f([x+\gamma]^{n})
\]
with inverse 
\[
(t^b_{n,\gamma})^{-1}:\ell^2(\Lambda_{k_n},\C^\fd)\to \ell^2(\Lambda_{k_n},\C^\fd)\,,\quad f \mapsto ((t^b_{n,\gamma})^{-1} f)(x)  = \e^{\i b (x_1-\gamma_1)\gamma_2}f([x-\gamma]^n)\,.
\]
Now the magnetic Bloch--Floquet transform maps to the   Brillouin zone $\T^2_{q,n}:=  \left(\frac{2\pi}{L}I_L\right)\times \left(\frac{2\pi}{ L} I_{L/q} \right) $, where, as above,  $I_m:=\{0,\ldots,m-1\}$:
\[
\mathcal{F}_{q,n}: \ell^2(\Lambda_{k_n};\C^{\fd}) \to \ell^2(\T^2_{q,n} \times I_q;\C^{\fd})\,,\quad f\mapsto \hat f_j(k) := \sum_{x\in\Z^2_q\cap \Lambda_{k_n}}\e^{\i k x} \e^{\i k_2 j} f([x+(0,j)]^n)\,,
\]
with inverse
\[
\mathcal{F}^{-1}_{q,n}:     \ell^2(\T^2_{q,n} \times I_q;\C^{\fd})\to \ell^2(\Lambda_{k_n};\C^{\fd})\,,\quad f \mapsto \check f(x) := \frac{1}{|\T^2_{q,n}|}\sum_{k\in\T^2_{q,n}} \e^{-\i k x} f_{[x]_q}(k)\,.
\]
For $x \in\Z^2_q$ and $y\in\Z^2$ we find that
\begin{eqnarray*}
A_n([x+(0,j)]^n,y) 
    &=& \sum_{z\in (L\Z)^2} A([x+(0,j)]^n,y+z)
    \\
    &=& \sum_{z\in (L\Z)^2} A([x+(0,j)]^n + (x+(0,j))^n,y+z+ (x+(0,j))^n)
     \\
    &=& \sum_{z\in (L\Z)^2} A(x+(0,j),y+z+ (x+(0,j))^n)
     \\
    &=& \sum_{w\in (L\Z)^2} A(x+(0,j),y+w)
    \\
     &=& \sum_{w\in (L\Z)^2} A(x-y_q-w+(0,j),[y]_q)
      \\
     &=& \sum_{w\in (L\Z)^2} a_{j[y]_q}(x-y_q-w)
\end{eqnarray*}
and thus
\begin{align*}
    \hspace{0em}&\hspace{-0em} 
    (\mathcal{F}_{q,n}A_n\mathcal{F}^{-1}_{q,n} \hat f)(k,j)
    \\
    &=
    \tfrac{1}{|\T^2_{q,n}|}\sum_{x\in\Z^2_q\cap\Lambda_{k_n}}\!\!
    \e^{\i k x} \e^{\i k_2 j} \sum_{y\in\Lambda_{k_n}} A_n([x+(0,j)]^n,y) \sum_{t\in\T^2_{q,n}} \e^{-\i t y} \hat f_{[y]_q}(t)
    \\
    &=
    \tfrac{1}{|\T^2_{q,n}|}\sum_{x\in\Z^2_q\cap\Lambda_{k_n}}\!\!
    \e^{\i k x}\e^{\i k_2 j} \sum_{y\in\Lambda_{k_n}} \sum_{w\in(L\Z)^2} a_{j[y]_q}(x-y_q-w)\sum_{t\in\T^2_{q,n}} \e^{-\i t y} \hat f_{[y]_q}(t)
    \\
    &=
    \tfrac{1}{|\T^2_{q,n}|}\sum_{x\in\Z^2_q\cap\Lambda_{k_n}}\!\!
    \e^{\i k x} \e^{\i k_2 j}\!\! \sum_{y\in\Z^2_q\cap \Lambda_{k_n}  }\sum_{l\in I_q} \sum_{w\in(L\Z)^2} a_{jl}(x-y-w)\sum_{t\in\T^2_{q,n}} \e^{-\i t (y+(0,l))} \hat f_{l}(t)
    \\
    &=
    \tfrac{1}{|\T^2_{q,n}|}\sum_{x\in\Z^2_q\cap\Lambda_{k_n}}\!\!
    \e^{\i k x} \e^{\i k_2 j}\!\! \sum_{y\in\Z^2_q\cap\Lambda_{k_n}} \sum_{l\in I_q}   \sum_{w\in(L\Z)^2}a_{jl}(x-(y+w))\sum_{t\in\T^2_{q,n}} \e^{-\i t (y+w+(0,l)) } \hat f_l(t)
    \\ 
    &=
    \tfrac{1}{|\T^2_{q,n}|}\sum_{x\in\Z^2_q\cap\Lambda_{k_n}}\!\!
    \e^{\i k x} \e^{\i k_2 j} \sum_{y\in\Z^2_q} \sum_{l\in I_q}   a_{jl}(x-y)\sum_{t\in\T^2_{q,n}} \e^{-\i t (y+(0,l)) } \hat f_l(t)
    \\
    &=
    \tfrac{1}{|\T^2_{q,n}|}\sum_{t\in\T^2_{q,n}} \sum_{z\in\Z^2_q} 
    \underbrace{\sum_{x\in\Z^2_q\cap\Lambda_{k_n}}\e^{\i  x(k-t)}}_{=|\T^2_{q,n}|\delta(k-t)}
    \e^{\i (k_2 j-t_2l)} \sum_{l\in I_q}   a_{jl}(z)\e^{ \i t  z } \hat f_l(t)
    \\
    &=
    \sum_{l\in I_q} \sum_{z\in\Z^2_q}   a_{jl}(z)\e^{ \i k  (z+(0,j-l)) } \hat f_l(k)
    \\
    &=
    \sum_{l\in I_q} \hat a_{jl}(k)\,\hat f_l(k)\,.
\end{align*}
So the restriction acts fiber-wise exactly as the infinite volume operator.

For self-adjoint $\Z^2_q$-periodic operators $A$ that implies, in particular, that the restriction $(f(A))_n$ of any bounded Borel function $f$ of $A$ to $\Lambda_{k_n}$ equals $f(A_n)$. In particular,
we have \eqref{eq:fermi_restriction}. 

\subsection{Convergence of states and of cocycles} \label{sec:convergences}

\subsubsection{Proof of Lemma \ref{lem:finite_vol_approx}} \label{proof:finite_vol_approx}
    We define the non-periodic finite volume approximation of $\Phi^t$ as the interaction given by
    \begin{align*}
        \Phi^{t,k,\mathrm{np}}(M) \coloneq
        \left\{\begin{array}{cl}
        T_{\g}\E_{\La_{0}}\Phi^{t,k}_0 & \mbox{if  $M=\La_0 +\g$ for some $\g\in \Lambda_k$}\\
        T_{\g}(\E_{\La_{l}}\Phi^{t,k}_0 - \E_{\La_{l-1}}\Phi^{t,k}_0)&
        \mbox{if  $M=\La_l +\g$ for some $\g\in\Lambda_k$ and $l\in\N$}\\
        0 & \mbox{otherwise}
        \end{array}
        \right
        .
    \end{align*}
    and denote its generated cocycle of automorphisms by $(\alpha_{s,t}^{k,\mathrm{np}})_{s,t\in I}$. Similarly we define a coarse grained version of $\Phi^t$ by
    \begin{align*}
        \Phi^{t,\mathrm{cg}}(M) \coloneq
        \left\{\begin{array}{cl}
        T_{\g}\E_{\La_{0}}\Phi^{t}_0 & 
        \mbox{if  $M=\La_0 +\g$ for some $\g\in \Z^2$}\\
        T_{\g}(\E_{\La_{l}}\Phi^{t}_0 - \E_{\La_{l-1}}\Phi^{t}_0)&
        \mbox{if  $M=\La_l +\g$ for some $\g\in \Z^2$ and $l\in\N$}\\
        0 & \mbox{otherwise}
        \end{array}
        \right
        .
    \end{align*}
    and note that it induces the same derivation on $\mA$ as $\Phi^t$ and therefore generates the same cocycle of automorphisms $(\alpha_{s,t})_{s,t\in I}$. A straightforward computation shows that for  $\nu \in \N_0$ there is a constant $c_\nu$ such that for all $k\in \N_0$ we have $\lVert \Phi^{t,k,\mathrm{np}} \rVert_\nu \leq c_\nu\,\lVert \Phi_0^{t,k} \rVert_{\nu+4}$ and $\lVert \Phi^{t,\mathrm{cg}} \rVert_\nu \leq c_\nu\,\lVert \Phi_0^{t} \rVert_{\nu+4}$, which gives us a bound for the interaction norms that is uniform in $k$.
    Now by  \cite[Lemma 2.10]{BTW26automorphic}, we have that for each $\nu \in \N_0$ there exists an increasing function $f_\nu \colon [0,\infty) \to (0,\infty)$ that grows at most polynomially at infinity and satisfies
    \begin{align*}
        \norm{\alpha_{s,t} \, A}_\nu \leq f_\nu(\abs{t-s}) \, \norm{A}_\nu
    \end{align*}
    and
    \begin{align*}
        \norm{\alpha^{k,\mathrm{np}}_{s,t} \, A}_\nu \leq f_\nu(\abs{t-s}) \, \norm{A}_\nu
    \end{align*}
    for all $k\in \N_0$, $s,t\in I$ and $A\in \mA_\infty$.
    
    Let $k\in \N_0$, $s,t\in I$ and $A\in \mA_\infty$. By the locality of $(\alpha_{u,v})_{u,v\in I}$ we find 
    \[
        \norm{\alpha_{s,t}(A- \E_{\Lambda_k}A)}_\nu
        \leq f_\nu(|s-t|) \norm{A- \E_{\Lambda_k}A}_\nu 
        \,\leq\, 
        2\,f_\nu(|s-t|)  \frac{\norm{A}_{\nu+1}}{(1+k)}
    \]
    which shows that we can assume   $A=\E_{\Lambda_k}A$ for the rest of the proof.
     
    We now proceed by adding and subtracting the non-periodic approximation in the expression and bounding the two resulting terms separately:
    \begin{align*}
        \norm{\alpha_{s,t}\, A - \alpha_{s,t}^{k} \, A}_\nu
        \leq
        \norm{\alpha_{s,t}\, A - \alpha_{s,t}^{k,\mathrm{np}} \, A}_\nu 
        +
        \norm{\alpha_{s,t}^{k,\mathrm{np}}\, A - \alpha_{s,t}^{k} \, A}_\nu
        .
    \end{align*}
    In the second term $\alpha_{s,t}^{k,\mathrm{np}}\, A - \alpha_{s,t}^{k} \, A$ lies in $\mA_{\Lambda_{2k}}$ and therefore 
    \begin{align*}
        \norm{\alpha_{s,t}^{k,\mathrm{np}}\, A - \alpha_{s,t}^{k} \, A}_\nu
        \leq
        3\,(1+2k)^\nu \, \norm{\alpha_{s,t}^{k,\mathrm{np}}\, A - \alpha_{s,t}^{k} \, A}
        .
    \end{align*}
    Now the fundamental theorem of calculus implies
    \begin{align*}
        \hspace{2em}&\hspace{-2em}
        \norm{\alpha_{s,t}^{k,\mathrm{np}}\, A - \alpha_{s,t}^{k} \, A}_\nu
        \\
        &\leq
        3\,(1+2k)^\nu \int_{s\land t}^{s\lor t} \dd u \, \sum_{\gamma \in \Lambda_k} \norm{[(T_\gamma - T_\gamma|_{\Lambda_k})\,  \Phi_0^{u,k}, \, \alpha_{u,t}^{k,\mathrm{np}} \, A] }
        \\
        &\leq 
        3\,(1+2k)^\nu \int_{s\land t}^{s\lor t} \dd u \, \sum_{\gamma \in \Lambda_k} \norm{[(T_\gamma - T_\gamma|_{\Lambda_k})\, \E_{\Lambda_{k/3}} \, \Phi_0^{u,k}, \, \alpha_{u,t}^{k,\mathrm{np}} \, A] }
        \\
        &\hspace{2em}+ 
        3\,(1+2k)^\nu \int_{s\land t}^{s\lor t} \dd u \, \sum_{\gamma \in \Lambda_k} \norm{[(T_\gamma - T_\gamma|_{\Lambda_k})\, (1 - \E_{\Lambda_{k/3}}) \, \Phi_0^{u,k}, \, \alpha_{u,t}^{k,\mathrm{np}} \, A] }
        \\
        & \eqcolon S_1 + S_2
    \end{align*}
    For $S_1$ we find
    \begin{align*}
        S_1
        &\leq 
        3\,(1+2k)^\nu \int_{s\land t}^{s\lor t} \dd u \, \sum_{\gamma \in \Lambda_{2k/3}} \norm{[(T_\gamma - T_\gamma|_{\Lambda_k})\, \E_{\Lambda_{k/3}} \, \Phi_0^{u,k}, \, \alpha_{u,t}^{k,\mathrm{np}} \, A] }
        \\
        &\hspace{1em}+
        3\,(1+2k)^\nu \int_{s\land t}^{s\lor t} \dd u \, \sum_{\gamma \in \Lambda_k \setminus \Lambda_{2k/3}} \norm{[(T_\gamma - T_\gamma|_{\Lambda_k})\, \E_{\Lambda_{k/3}} \, \Phi_0^{u,k}, \, \E_{\Lambda_{k/3}} \, \alpha_{u,t}^{k,\mathrm{np}} \, A] }
        \\
        &\hspace{1em}+
        3\,(1+2k)^\nu \int_{s\land t}^{s\lor t} \dd u \, \sum_{\gamma \in \Lambda_k \setminus \Lambda_{2k/3}} \norm{[(T_\gamma - T_\gamma|_{\Lambda_k})\, \E_{\Lambda_{k/3}} \, \Phi_0^{u,k}, \, (1-\E_{\Lambda_{k/3}}) \, \alpha_{u,t}^{k,\mathrm{np}} \, A] }
        ,
    \end{align*}
    where the first term vanishes, since the two translations act identically and the second term vanishes, since the operators in the commutator are supported on disjoint sets with $\Phi_0^{u,k}\in \mA^N$. For the remaining part we find
    \begin{align*}
        S_1
        &\leq
        3\,(1+2k)^\nu \int_{s\land t}^{s\lor t} \dd u \, \sum_{\gamma \in \Lambda_k \setminus \Lambda_{2k/3}} 4\,\norm{\Phi_0^{u,k} } \, \norm{ (1-\E_{\Lambda_{k/3}}) \, \alpha_{u,t}^{k,\mathrm{np}} \, A }
        \\
        &\leq
        3\,(1+2k)^\nu \int_{s\land t}^{s\lor t} \dd u \, \sum_{\gamma \in \Lambda_k \setminus \Lambda_{2k/3}} 4\,\norm{\Phi_0^{u,k} } \, \frac{\norm{\alpha_{u,t}^{k,\mathrm{np}} \, A}_{\nu+3}}{(1+k/3)^{\nu+3}}
        \\
        &\leq
         \frac{12\,(1+2k)^{\nu+2}}{(1+k/3)^{\nu+3}} \, \abs{t-s}\, \sup_{u\in [s \land t, s\lor t]} \norm{\Phi_0^{u,k} } \, f_{\nu+3}(\abs{t-u}) \,\norm{ A}_{\nu+3} , 
    \end{align*}
    which is a bound of the required form. Similarly, $S_2$ can be bounded as follows
    \begin{align*}
        S_2 
        &\leq
        12\,(1+2k)^{\nu+2} \int_{s\land t}^{s\lor t} \dd u \,   \norm{(1 - \E_{\Lambda_{k/3}}) \, \Phi_0^{u,k}} \, \norm{A }
        \\
        &\leq
        \frac{12\,(1+2k)^{\nu+2}}{(1+k/3)^{\nu + 3}} \, \abs{t-s}\,\sup_{u\in [s \land t, s\lor t]} \norm{\Phi_0^{u,k}}_{\nu+3} \, \norm{A}.
    \end{align*}
    It remains to show that $\norm{\alpha_{s,t}\, A - \alpha_{s,t}^{k,\mathrm{np}} \, A}_\nu$ admits the same kind of bound. Again, the fundamental theorem of calculus  yields
    \begin{align*}
    \norm{\alpha_{s,t} \, A - \alpha_{s,t}^{k, \mathrm{np}} \, A }_{\nu} 
        &\leq 
        \int_{s\land t}^{s\lor t} \dd u \, \norm{\alpha_{s,u} \,  (\i\mL_{\Phi^u} - \i\mL_{\Phi^{u,k,\mathrm{np}}}) \, \alpha_{u,t}^{k,\mathrm{np}} \, A }_\nu
        \\
        &\leq
        \int_{s\land t}^{s\lor t} \dd u \, f_\nu(\abs{u-s}) \, \norm{  (\i \mL_{\Phi^u} - \i \mL_{\Phi^{u,k,\mathrm{np}}}) \, \alpha_{u,t}^{k,\mathrm{np}} \, A }_\nu \numberthis \label{eq:comparison-nonperiodic}
        .
    \end{align*}
    Next we write the difference of the generators as a sum, which we then split into two parts:
    \begin{align*}
        \hspace{2em}&\hspace{-2em}
        \norm{  (\i \mL_{\Phi^u} - \i \mL_{\Phi^{u,k,\mathrm{np}}}) \, \alpha_{u,t}^{k,\mathrm{np}} \, A }_\nu
        \\
        &\leq
        \sum_{\gamma \in \Z^2\setminus \Lambda_k} \norm{[T_\gamma \, \Phi_0^u, \, \alpha^{k,\mathrm{np}}_{u,t}\, A ]}_\nu 
        + \sum_{\gamma \in \Lambda_k} 
        \norm{
            [
                T_\gamma \, (\Phi_0^u - \Phi^{u,k}_0)
            ,\, 
                \alpha^{k,\mathrm{np}}_{u,t}\, A
            ]
        }_\nu
        \numberthis \label{eqn:splitSj}
        \\
        &\eqcolon
        S_3 + S_4 .
    \end{align*}
    We again estimate the two terms separately, beginning with
    \begin{align*}
        S_3 
        &\leq 
        \sum_{\gamma \in \Z^2\setminus \Lambda_k} 4^{\nu+7} \, \frac{\norm{\Phi^u_0}_{\nu+4} \, \norm{\alpha^{k,\mathrm{np}}_{u,t}\, A}_{\nu+4}}{(1+\norm{\gamma})^{4}}
        \\
        &\leq 
        4^{\nu+7}\, \frac{f_{\nu+4}(\abs{t-u}) \,   \norm{A}_{\nu+4}}{1+k} \sum_{\gamma \in \Z^2} \frac{\sup_{v\in I} \norm{\Phi^v_0}_{\nu+4}}{(1+\norm{\gamma})^3} 
        .
    \end{align*}
    Next
    \begin{align*}
        S_4 
        &\leq
        \sum_{\gamma \in  \Lambda_k} 4^{\nu+3} \, \frac{\norm{\Phi^u_0- \Phi^{u,k}_0}_{\nu+3} \, \norm{\alpha^{k,\mathrm{np}}_{u,t}\, A}_{\nu+3}}{(1+\norm{\gamma})^{3}}
        \\
        &\leq
        f_{\nu+3}(\abs{t-u}) \,\norm{A}_{\nu+3} \sum_{\gamma \in  \Z^2} 4^{\nu+3} \, \frac{\sup_{v\in I}\norm{\Phi^v_0- \Phi^{v,k}_0}_{\nu+3}  }{(1+\norm{\gamma})^{3}}
        ,
    \end{align*}
    where the sum goes to $0$ as $k \to \infty$ due to dominated convergence. After estimating the integral in \eqref{eq:comparison-nonperiodic} with a supremum, this gives us the desired bound on $\norm{\alpha_{s,t} \, A - \alpha_{s,t}^{k} \, A }_{\nu}$.

    To show the bound for $\norm{(\i\mL_{\Phi^t} - \i\mL_{\Phi^{t,k}})\, A}_\nu$ we also insert the non-periodic approximation 
    \begin{align*}
        \norm{(\i\mL_{\Phi^t} - \i\mL_{\Phi^{t,k}})\, A}_\nu
        &\leq
        \norm{(\i\mL_{\Phi^t} - \i\mL_{\Phi^{t,k,\mathrm{np}}})\, A}_\nu + \norm{(\i\mL_{\Phi^{t,k,\mathrm{np}}} - \i\mL_{\Phi^{t,k}})\, A}_\nu
    \end{align*}
    and follow the same steps as above with $s=t$.
    \qed

\subsubsection{Proof of Lemma \ref{lem: convergence of cocycles}}\label{proof: convergence of cocycles}
    As in the proof of Lemma \ref{lem:finite_vol_approx}, for each $n\in \N_0\cup \{\infty\}$ we can define the coarse grained interaction 
    \begin{align*}
        \Phi^{t,n,\mathrm{cg}}(M) \coloneq
        \left\{\begin{array}{cl}
        T^n_{\g}\E_{\La_{0}}\Phi^{t,n}_0 & 
        \mbox{if  $M=\La_0 +\g$ for some $\g\in \Z^2$}\\
        T^n_{\g}(\E_{\La_{l}}\Phi^{t,n}_0 - \E_{\La_{l-1}}\Phi^{t,n}_0)&
        \mbox{if  $M=\La_l +\g$ for some $\g\in \Z^2$ and $l\in\N$}\\
        0 & \mbox{otherwise}
        \end{array}
        \right
        .
    \end{align*}
    which generates the same cocycle of automorphisms as $(\Phi^{t,n})_{t\in I}$ and satisfies $\norm{\Phi^{t,n,\mathrm{cg}}}_\nu \leq c\, \norm{\Phi^{t,n}_0}_{\nu+4}$ for all $\nu\in \N_0$ with some constant $c_\nu>0$ that is independent of $n$.
    Due to \cite[Lemma 2.10]{BTW26automorphic}, this implies that for each $\nu \in \N_0$ there exists an increasing function $f_\nu \colon [0,\infty) \to (0,\infty)$ that grows at most polynomially at infinity and satisfies
    \begin{align*}
        \norm{ \alpha_{(s,t),n} \, A }_\nu \leq f_\nu(\abs{t-s})\, \norm{ A }_\nu 
    \end{align*}
    for all $s,t\in I$, $A\in \mA_\infty$ and $n \in \N_0 \cup \{\infty\}$. Now let $n,\nu\in \N_0$, $s,t\in I$ and $A\in \mA_\infty$. By the fundamental theorem of calculus we find
    \begin{align*}
        \hspace{1em}&\hspace{-1em}
        \lVert \alpha_{(s,t),n} \, A - \alpha_{(s,t),\infty} \, A \rVert_\nu 
        \\
        &\leq 
        \lVert \int_{s }^{t} \dd u \,  \alpha_{(s,u),n} \, (\i\mL_{\Phi^{u,n}} - \i\mL_{\Phi^{u,\infty}}) \, \alpha_{(u,t),\infty} \, A \rVert_\nu 
        \\
        &\leq
        \lvert t- s \rvert \, f_\nu(\abs{t-s}) \, \sup_{u\in [s \land t, s\lor t]} \lVert (\mL_{\Phi^{u,n}} - \mL_{\Phi^{u,\infty}}) \, \alpha_{(u,t),\infty} \, A \rVert_\nu
        \\
        &\leq
        \abs{t- s} \, f_\nu(\abs{t-s})   \sum_{\g \in \Z^2} \sup_{u\in [s \land t, s\lor t]} \lVert[\, T^n_\gamma \, \Phi^{u,n}_0 - T_\g^\infty \,  \Phi^{u,\infty}_0 , \, \alpha_{(u,t),\infty} \, A\, ] \rVert_\nu
        \\
        &\leq
        4^{\nu+6}\, \abs{t- s} \, f_\nu(\abs{t-s}) \sum_{\g \in \Z^2} \sup_{u\in [s \land t, s\lor t]}\!\! \frac{\lVert \Phi^{u,n}_0 - (T_\g^n)^{-1}\,T_\g^\infty \,  \Phi^{u,\infty}_0 \rVert_{\nu+3} \, \lVert \alpha_{(u,t),\infty} \, A \rVert_{\nu+3}}{(1+\lVert \g \rVert)^3}
        \\
        &\leq
        4^{\nu+6}\, \abs{t- s} \, f_\nu(\abs{t-s}) \, f_{\nu+3}(\abs{t-s}) \,  \lVert  \, A \rVert_{\nu+3} \!\sum_{\g \in \Z^2}  \sup_{u\in I}\frac{ \lVert \Phi^{u,n}_0 - (T_\g^n)^{-1}\,T_\g^\infty \,  \Phi^{u,\infty}_0 \rVert_{\nu+3} }{(1+\lVert \g \rVert)^3}
        .
    \end{align*}
    The sum converges to zero as $n\to \infty$ due to dominated convergence, where we obtain a summable upper bound from 
    \begin{align*}
        \hspace{2em}&\hspace{-2em}
       \sup_{n\in\N}  \sup_{\gamma\in\Z^2}\sup_{u\in I}\,\norm{\Phi^{u,n}_0 -(T_\g^n)^{-1}\,T_\g^\infty\, \Phi^{u,\infty}_0}_{\nu+3}
       \\
       &\leq
       \sup_{n\in\N} \sup_{u\in I}\,\norm{\Phi^{u,n}_0}_{\nu+3} + \sup_{u\in I}\,\norm{\Phi^{u,\infty}_0}_{\nu+3} <\infty\,.
    \end{align*}
    Here we use that $\Phi_0^{u,n}\stackrel{n\to\infty}{\rightarrow}\Phi_0^{u,\infty}$ in $\mA_{\nu+3}$
     and the fact that $\norm{(T_\g^n)^{-1}\,T_\g^\infty \, \Phi^{u,\infty}_0}_\nu = \norm{\Phi^{u,\infty}_0}_\nu$. The latter statement follows from the fact that the automorphism $(T_\g^n)^{-1}\,T_\g^\infty$ does not change the support of operators, which allows us to conclude $ \E_M \, (T_\g^n)^{-1}\,T_\g^\infty = (T_\g^n)^{-1}\,T_\g^\infty \, \E_M$ for all $M\in P_0(\Z^2)\,$  from the uniqueness in Proposition \ref{Ex+UniqueExpectation}.
    
    The bound 
    \begin{align*}
        \sup_{t\in I} \sup_{A\in \mA_\infty\setminus\{0\}} \frac{\norm{\mL_{\Phi^{t,n}} A - \mL_{\Phi^t} A }_{\nu}}{\norm{A}_{\nu+3}} 
        \leq
        4^{\nu+6}\,  \sum_{\g \in \Z^2}  \frac{\sup_{u\in I} \lVert \Phi^{u,n}_0 - (T_\g^n)^{-1}\,T_\g^\infty \,  \Phi^{u,\infty}_0 \rVert_{\nu+3} }{(1+\lVert \g \rVert)^3}
        ,
    \end{align*}
    which can be obtained by following the same steps and together the two bounds imply the claimed statements.
\qed

\subsubsection{Proof of Lemma \ref{lem: translating back and forth}}\label{proof:translating back and forth}
    Since automorphisms on the CAR algebra are determined by their action on the annihilation operators it suffices to show the identity only for them. Let $y\in \Z^2$ and $j\in \{1,\dots, \fd\}$. For the left-hand side we find
    \begin{align*}
        (T^b_\g)^{-1}\, T_\g^{b_0} \, a_{y,j}
        &= (T^b_\g)^{-1}\,\e^{-\i b_0 y_1 \g_2}\, a_{y+\g,j}
        \\
        &= \e^{-\i b_0 y_1 \g_2} \, \e^{\i b (y+\g-\g)_1 \g_2}\, a_{y,j} 
        = \e^{\i (b-b_0) y_1 \g_2} \, a_{y,j}
        .
    \end{align*}
    For the right-hand side one can now compare with the above value
    \begin{align*}
        \hspace{2em}&\hspace{-2em}
        \left\| \left( \e^{-\i (b-b_0) \g_2 \mL_{X_1}}  - \e^{\i (b-b_0) y_1 \g_2} \right)  a_{y,j}\right\|
        \\
        &\leq
        \int_{b \land b_0}^{ b \lor b_0} \dd s   \, \left\| \e^{-\i (s-b_0) \g_2 \mL_{X_1}} \, (-\i\, \g_2 \, \mL_{X_1} - \i  \, y_1 \, \g_2 ) \, \e^{\i (b-s) y_1 \g_2} \, a_{y,j}\right\|
        \\
        &\leq
        \abs{b-b_0}   \, \abs{\g_2} \, \norm{  \mL_{X_1}\, a_{y,j} +  y_1 \,  a_{y,j}} = 0
        .
    \end{align*}
    To show the bound, we first note that, since $X_1$ is an on-site interaction, its generated group of automorphisms $s \mapsto \e^{\i s  \mL_{X_1}}$ does not change the support of operators. Hence, for all $M\in P_0(\Z^2)$ and $s\in \R$ it holds that $\E_M \, \e^{\i s  \mL_{X_1}} = \e^{\i s  \mL_{X_1}} \, \E_M$, and therefore, for all $\nu \in \N_0$ and $A \in  \mA_\infty$, one finds $\norm{\e^{\i s  \mL_{X_1}}\, A}_\nu = \norm{A}_\nu$. Now a simple Duhamel argument yields
    \begin{align*}
        \norm{\e^{-\i (b-b_0) \gamma_2 \mL_{X_1}} \, A -A}_\nu 
        &\leq
        \int_{b\land b_0}^{b \lor b_0} \dd s \,  \norm{\e^{-\i (s-b_0) \gamma_2\mL_{X_1}} \, (-\i) \, \g_2\, \mL_{X_1}\, A }_\nu
        \\
        &\leq
        \abs{b-b_0}\, \abs{\g_2}\,  \norm{\mL_{X_1}\, A }_\nu
        \\
        &\leq \abs{b-b_0}\, \abs{\g_2}\, \sum_{x\in \Z^2} 4^{7+\nu} \, \frac{\abs{x_1} \, \norm{n_0}_{\nu+4} \,\norm{A}_{\nu+4}}{(1+\norm{x})^4}
        ,
    \end{align*}
    where we used \cite[Lemma A.1]{WMMMT24} in the last step.
\qed

\printbibliography

\end{document}